\documentclass[11pt]{article}

\usepackage{amsthm,amsmath}

\usepackage{natbib}
\usepackage{amssymb,amsfonts,bm,enumerate,graphicx}
\usepackage{xcolor,color}
\usepackage{booktabs}
\usepackage{caption}
\usepackage{subcaption}
\usepackage{float}
\usepackage{enumitem}
\usepackage{algorithmicx}
\usepackage{algorithm}
\usepackage{algpseudocode}
\usepackage{url}
\usepackage{fancyhdr}
\usepackage{soul}  
\usepackage[margin=1in]{geometry} 
\usepackage{setspace}
\usepackage{lineno}
\usepackage{multirow}
\pdfoutput=1
\providecommand{\keywords}[1]{\textrm{KEY WORDS: } #1}
\newcommand{\bs}{\bm{s}}
\def\bs{{\boldsymbol{s}}}

\def\pr{{\rm Pr}}

\def\iid{{\ {\buildrel \rm{iid}\over \sim}\ }}
\def\calS{{\mathcal{S}}}
\def\E{{\rm E}}
\DeclareMathSymbol{\Real}{\mathbin}{AMSb}{"52}
\DeclareMathSymbol{\Natural}{\mathbin}{AMSb}{"4E}

\theoremstyle{plain}
\newtheorem{theorem}{Theorem}[section]
\newtheorem{proposition}[theorem]{Proposition}

\title{A Hierarchical Max-Infinitely Divisible Spatial Model for Extreme Precipitation}
\date{}
\author{Gregory P. Bopp\thanks{Department of Statistics, Pennsylvania State University, 326 Thomas Building, University Park, PA 16802, U.S.A. Emails: gxb951@psu.edu}, \, Benjamin A. Shaby\thanks{Department of Statistics, Colorado State University, 211 Statistics Building, Fort Collins, CO 80523.  Email: bshaby@colostate.edu} \, and Rapha\"{e}l Huser\thanks{Computer, Electrical and Mathematical Sciences and Engineering (CEMSE) Division, King Abdullah University of Science and Technology (KAUST), Thuwal 23955-6900, Saudi Arabia. Email: raphael.huser@kaust.edu.sa}}

\begin{document}
\maketitle
\begin{abstract}
Understanding the spatial extent of extreme precipitation is necessary for determining flood risk and adequately designing infrastructure (e.g., stormwater pipes) to withstand such hazards. While environmental phenomena typically exhibit weakening spatial dependence at increasingly extreme levels, limiting max-stable process models for block maxima have a rigid dependence structure that does not capture this type of behavior. We propose a flexible Bayesian model from a broader family of (conditionally) max-infinitely divisible processes that allows for weakening spatial dependence at increasingly extreme levels, and due to a hierarchical representation of the likelihood in terms of random effects, our inference approach scales to large datasets. Therefore, our model not only has a flexible dependence structure, but it also allows for fast, fully Bayesian inference, prediction and conditional simulation in high dimensions. The proposed model is constructed using flexible random basis functions that are estimated from the data, allowing for straightforward inspection of the predominant spatial patterns of extremes. In addition, the described process possesses (conditional) max-stability as a special case, making inference on the tail dependence class possible. We apply our model to extreme precipitation in North-Eastern America, and show that the proposed model adequately captures the extremal behavior of the data. Interestingly, we find that the principal modes of spatial variation estimated from our model resemble observed patterns in extreme precipitation events occurring along the coast (e.g., with localized tropical cyclones and convective storms) and mountain range borders. Our model, which can easily be adapted to other types of environmental datasets, is therefore useful to identify extreme weather patterns and regions at risk.
\end{abstract}

\keywords{max-infinitely divisible process; max-stable process;
sub-asymptotic extremes; block maxima.}

\section{INTRODUCTION} 
\label{sec:introduction}

The risk of precipitation-induced flooding (pluvial flooding) is strongly determined by the spatial extent of severe storms, and therefore, there is a need to adequately describe the spatial dependence properties of extreme precipitation. With this goal in mind, we propose a scalable model for spatial extremes that relaxes the rigid dependence structure of asymptotic max-stable models, characterizes the main modes of spatial variability using interpretable spatial factors, and allows for easy prediction at unobserved locations. The areal aspect of extreme precipitation plays a role in flood risk assessment. Precipitation falling over a single drainage basin flows into a common outlet, the aggregate effects of which can be devastating in large volumes. In 2006, heavy precipitation over the Susquehanna River basin in New York and Pennsylvania caused record high discharges along the Susquehanna River and flooding in the region, ultimately leading to federal-level disaster declarations and disaster-recovery assistance from the US Federal Emergency Management Agency (FEMA) in excess of \$227 million \citep{Suro09}. 

The last decade has seen a considerable amount of research on the spatial dependence modeling of extremes, in part because of the hazard that extreme weather events pose to human life and property. For recent reviews, see \cite{Davison12, Davison13, Davison19} and \cite{Davison15}. The classical geostatistical Gaussian process models that are ideal for modeling the bulk of a distribution have weak tail-dependence and do not enforce the specific type of positive dependence structure inherent to extremes. Two classes of models, max-stable processes \citep{deHaan07} and generalized Pareto processes \citep{Ferreira14,Thibaud15}, have proven to be useful tools for the modeling of spatial extremes. Max-stable process models are infinite-dimensional generalizations of the limiting models for componentwise maxima. They are asymptotically justified models for pointwise maxima over an infinite collection of independent processes after suitable renormalization, a property which has made them prime candidates for the modeling of spatial extremes. In practice, maxima are taken over large, but finite blocks (e.g., months, years). An approximation error is incurred when applying limiting models to pointwise maxima over finite blocks, and the degree of this error will depend on the rate of convergence of the modeled process as the block size grows. Furthermore, the approximation error is more pronounced when the observed process exhibits weakening spatial dependence at increasingly high quantiles, as the spatial dependence of limiting max-stable processes is the same across all levels of the distribution, and hence would overestimate the level of dependence in the data. For more discussion, see, e.g., \cite{Wadsworth12}. Empirical evidence has shown that environmental processes often exhibit weakening spatial dependence at more extreme levels, which has led some to consider non-limiting models for flexible tail dependence modeling \citep{Morris16,Huser17,Huser18, HuserWadsworth19}. In particular, \cite{Morris16} use a random partition of their spatial domain and locally defined, asymptotically dependent skew-t processes to induce long-range asymptotic independence but short-range asymptotic dependence. 

In this paper, we aim to extend a class of max-stable models in order to flexibly capture spatial dependence characteristics for sub-asymptotic block maxima data, while still retaining the positive dependence structure inherent to distributions for maxima. The general class of models that we consider, which nests the class of max-stable models, are known as max-infinitely divisible (max-id) processes \citep[][Chapter 5]{Resnick87}. Suppose a random vector $\mathbf{X}$ has joint distribution $F_X$, then the distribution of maxima of $n$ independent and identically distributed (i.i.d.) replicates $\mathbf{X}_1,\ldots, \mathbf{X}_n$, taken componentwise, has distribution function $F_X^n$. The max-id property applies to the converse statement. Suppose that $\mathbf{Z}$ is a random vector of componentwise maxima, composed from a collection of $n$ i.i.d.\ vectors. Then if $\mathbf{Z}$ has distribution function $G$, there exists some root distribution $F$ such that $G(\mathbf{z}) = F^n(\mathbf{z})$, or equivalently such that $G^{1/n}(\mathbf{z}) = F(\mathbf{z})$. By continuous extension of the relation $G^{q/r} = F$ for $q,r \in \mathbb{N}$, we say that a distribution $G$ is max-id if and only if $G^s$ is a valid distribution for all real $s >0$. This is always the case for univariate distributions, but may not necessarily be so for multivariate distributions. Informally, max-id distributions are those which arise from taking componentwise maxima of i.i.d.\ random vectors and are therefore an appropriate class to constrain ourselves to if the goal is to model componentwise maxima. By slight abuse of language, we say that a spatial process is max-id if all its finite-dimensional distributions are max-id. Necessary and sufficient conditions for max-infinitely divisibility of a distribution function in $\mathbb{R}^2$ were first given by \cite{Balkema77}. More recently, mixing conditions for stationary max-id processes were explored by \cite{Kabluchko10}, and minimality of their spectral representations were described in \cite{Kabluchko16}.

Unlike limiting max-stable process models, which have a rigid spatial dependence structure, sub-families of the broader class of max-id processes do not impose such constraints and can accommodate different spatial dependence characteristics across various levels of a distribution \citetext{see, e.g. \citealp{Padoan13}, \citealp{Huser18}}. It is the lack of this feature that can cause max-stable processes to fit poorly, as many  processes of interest may exhibit spatial dependence at extreme but finite levels. Extrapolation of max-stable fits to higher quantiles in this scenario can cause overestimation of the risk of concurrent extremes \citep{Davison13}. Furthermore, the challenge of performing conditional simulation from max-stable models given observed values at many locations is a limiting factor for their use in practice \citep{Dombry13}. The Bayesian model that we develop in the remainder of the paper permits a conditional, hierarchical representation in terms of random effects that facilitates fast conditional simulation, which is useful for prediction at unobserved locations, and for handling missing values.

\section{HIERARCHICAL CONSTRUCTION OF SPATIAL MAX-ID MODELS}
\label{sec:hierarchical_construction_of_spatial_maxid}

\subsection{Max-Stable Reich and Shaby (2012) Model} 
\label{sub:Reich_Shaby}
Our proposed approach is an extension of the Bayesian hierarchical model developed by \citet{Reich12}, which we review here. The \cite{Reich12} model possesses the max-stability property while being tractable in high-dimensions due to its conditional representation in terms of positive-stable variables \citetext{see also \citealp{Fougeres09} and \citealp{Stephenson09}}. Let $\alpha\in(0,1)$ and consider a set of independent $\alpha$-stable random variables $A_1,\ldots,A_L\iid {\rm PS}(\alpha)$, where generically the Laplace transform of $A \sim {\rm PS}(\alpha)$ has the form: $\E\{\exp(-sA)\}=\exp(-s^\alpha),\, s \geq 0$. Then we construct the spatial process $Z(\bs)$ as the product of two independent processes,
\begin{equation}
Z(\bs)=\varepsilon(\bs)Y(\bs),
\label{eq:z_process}
\end{equation}
where $\varepsilon(\bs)$ is a white noise process (i.e., an everywhere-independent multiplicative nugget effect) with $(1/\alpha)$-Fr\'echet marginals, $\pr\{\varepsilon(\bs)\leq z\}=\exp(-z^{-1/\alpha})$, and $Y(\bs)$ is a spatially dependent process defined as an $L^p$-norm (for $p = 1/\alpha$) of scaled, spatially-varying basis functions $K_l(\bs)\geq0$, $l = 1,\ldots, L$:
\begin{equation}
Y(\bs)= \left\{\sum_{l=1}^L A_l K_l(\bs)^{1/\alpha}\right\}^{\alpha}.
\label{eq:residual_dep_proc}	
\end{equation}

The white noise process $\varepsilon(\bs)$ functions as a nugget effect, and accounts for measurement error occurring independently of the underlying process of interest. For small $\alpha$, the contribution of $Y(\bs)$ dominates that of the nugget effect, and vice-versa for large $\alpha$.

\cite{Reich12} used fixed, deterministic spatial basis functions. In other words, they assumed a Dirac prior on the space of valid basis functions, based on the following construction: let $\boldsymbol{v}_1,\ldots,\boldsymbol{v}_L\in\calS\in\Real^p$ be a collection of spatial knots over our spatial domain of interest $\calS$, and $K_l(\bs)$, $l=1,\ldots,L$, be Gaussian densities centered at each knot $\boldsymbol{v}_l$, normalized such that $\sum_{l=1}^{L}K_l(\bs)=1$ for all $\bs\in\calS$. The Gaussian density basis functions may be replaced with normalized functions from a much broader class while still giving a valid construction for $Y(\bs)$ in \eqref{eq:residual_dep_proc}. A more flexible prior for the kernels $K_l(\bs)$, $l=1,\ldots,L$, is discussed in Section~\ref{sub:flexible_gaussian_process_factor_model}.

The process $\{Z(\bs), \bs \in \mathcal{S}\}$ has finite-dimensional distributions
\begin{equation}
\pr\{Z(\bs_1)\leq z_1,\ldots,Z(\bs_D)\leq z_D\}
=\exp\left(-\sum_{l=1}^L\left[\sum_{j=1}^D \left\{z_j/K_l(\bs_j)\right\}^{-1/\alpha} \right]^\alpha\right),\quad z_1,\ldots,z_D>0\label{jointdistrReichShaby}
\end{equation}
\citetext{see \citealp{Tawn90}}, which follows from the Laplace transform of an $\alpha$-stable variable. From \eqref{jointdistrReichShaby} and the sum-to-one constraint, the marginal distributions are unit Fr\'echet, i.e., for all $\bs\in\calS$,
\begin{equation*}
\pr\{Z(\bs)\leq z\}=\exp\left(-\sum_{l=1}^L\left[\left\{z/K_l(\bs)\right\}^{-1/\alpha} \right]^\alpha\right)=\exp\left\{-z^{-1}\sum_{l=1}^LK_l(\bs)\right\}=\exp\left(-z^{-1}\right), \quad z>0.
\end{equation*}
Max-stability follows from \eqref{jointdistrReichShaby} by checking that
\begin{align}
\label{maxstable}
\pr\{Z(\bs_1)\leq nz_1,\ldots,Z(\bs_D)\leq nz_D\}^n&=\pr\{Z(\bs_1)\leq z_1,\ldots,Z(\bs_D)\leq z_D\}.
\end{align}
The max-stability property of $Z(\bs)$ makes it suitable for modeling spatial extremes in scenarios of strong, non-vanishing upper tail dependence. In Section~\ref{sub:subasymptotic_modeling_maxid}, we propose a more general max-id model, which can better cope with weakening tail dependence.

Inference may be efficiently performed by taking advantage of the inherent hierarchical structure of the \cite{Reich12} model, noticing that the data are independent conditional on the latent variables $\{A_l\}_{l=1}^L$, and may be written in terms of the Fr\'{e}chet distribution with scale parameter $Y(\bs)>0$ and shape parameter $1/\alpha>0$:
\begin{equation}
Z(\bs)|A_1, \ldots, A_L \overset{indep}{\sim} \text{Fr\'{e}chet}(Y(\bs),1/\alpha),\\
\end{equation}
for all $\bs \in \mathcal{S}$; that is, $\pr\{Z(\bs)\leq z\mid A_1,\ldots,A_L\}=\exp[-\{z/Y(\bs)\}^{-1/\alpha}]$, $z>0$.

\subsection{Sub-Asymptotic Modeling Based on a Max-Infinitely Divisible Process}
\label{sub:subasymptotic_modeling_maxid}
Despite the appealing properties of the \cite{Reich12} model, its deterministic basis functions and its max-stability make it fairly rigid in practice. Max-id processes are natural, flexible, sub-asymptotic models, that extend the class of max-stable processes while still possessing desirable properties reflecting the specific positive dependence structure of maxima. From \eqref{maxstable}, we can see that max-stable processes are always max-id. Therefore, the former form a smaller subclass within the latter. 

The tail dependence class strongly determines how the probability of joint exceedances of a high threshold extrapolates to extreme quantiles. A random vector $(X_1, X_2)^\top$ with marginal distributions $F_1$ and $F_2$ is said to be asymptotically independent if $\pr\{F_1(X_1)> u\mid F_2(X_2)> u\}\to0$ as $u \rightarrow 1$, and asymptotically dependent otherwise \citep{Coles99}. We say that a spatial process \{$X(\bs), \bs \in \mathcal{S}\}$ is asymptotically independent if $X(\bs_1)$ and $X(\bs_2)$ are asymptotically independent for all $\bs_1, \bs_2\in \mathcal{S}$, $\bs_1 \neq \bs_2$. Max-stable processes are always asymptotically dependent (except in the case of complete independence) and, therefore, they lack flexibility to adequately capture the tail behavior of asymptotically independent data. In this section, we propose an asymptotically independent max-id model that possesses the max-stable \cite{Reich12} model on the boundary of its parameter space. Dependence properties are further detailed in Section~\ref{sub:dependence_properties}.

To extend the \cite{Reich12} model to a more flexible max-id formulation, we can change the distribution of the underlying random basis coefficients $\{A_l\}_{l=1}^L$. The heavy-tailedness of the ${\rm PS}(\alpha)$ distribution yields asymptotic dependence and, by construction, max-stability. To achieve asymptotic independence while staying within the class of max-id processes, we can consider a lighter-tailed, exponentially tilted, positive-stable distribution,
\begin{align}\label{ExtendedPS}
A_1,\ldots,A_L\iid {\rm H}(\alpha,\delta,\theta),\qquad \alpha\in(0,1),\delta>0,\theta\geq0,
\end{align}
which was first introduced by \citet{Hougaard86} and further studied by \cite{Crowder89}, and has Laplace transform
\begin{align}\label{LaplaceExtendedPS}
\E\left\{\exp\left(-sX\right)\right\}&=\exp\left[-{\delta\over\alpha}\left\{(\theta+s)^\alpha - \theta^\alpha\right\}\right],\qquad X\sim {\rm H}(\alpha,\delta,\theta).
\end{align}
 Denote the ${\rm PS}(\alpha)$ density by $f_{{\rm PS}}(x)$. The ${\rm H}(\alpha, \delta, \theta)$ density $f_{\rm H}$ may be expressed in terms of the positive-stable density $f_{{\rm PS}}$ as
\begin{equation}
f_{{\rm H}}(x) = {f_{{\rm PS}}\{x(\alpha/\delta)^{1/\alpha}\}({\alpha/\delta})^{1/\alpha} \exp(-\theta x)\over\exp(\delta \theta^\alpha/\alpha)}, \qquad x>0,
\label{eq:hougaard_den}
\end{equation}
for $\alpha \in (0,1)$, $\theta \geq 0$, and $\delta > 0$ \citep{Hougaard86}. An efficient algorithm for simulating from ${\rm H}(\alpha, \delta, \theta)$ is given by \cite{Devroye09}. A simple rejection sampler for the case when $\theta$ is not large is given in the Supplementary Material. When $\delta=\alpha$ and $\theta=0$, we recover the positive-stable distribution ${\rm PS}(\alpha)\equiv{\rm H}(\alpha,\alpha,0)$. The parameter $\alpha$ controls the tail decay, with smaller values of $\alpha$ corresponding to heavier-tailed distributions. Moreover, the density becomes increasingly concentrated around one as $\alpha\to1$. When $\theta >0$, the gamma distribution with shape $\delta$ and rate $\theta$ is obtained as $\alpha\to0$.  

Upon reparameterization in terms of $\alpha^\star=\alpha$, $\delta^\star=(\delta/\alpha)^{1/\alpha}$ and $\theta^\star=(\delta/\alpha)^{1/\alpha}\theta$, we see from \eqref{eq:hougaard_den} that $\delta^\star=(\delta/\alpha)^{1/\alpha}$ is a scale parameter, which does not affect the dependence structure of our new model. Therefore, in the remainder of this paper, we set $\delta = \alpha$ (i.e., $\delta^\star=1$) and use ${\rm H}(\alpha,\alpha,\theta)$ throughout without any loss in flexibility.

When $\delta = \alpha$ and $\theta >0$, $f_{\rm H}$ is an exponentially tilted form of $f_{{\rm PS}}$, where the parameter $\theta$ has the effect of exponentially tapering the tail of $f_{{\rm PS}}$ at rate $\theta$. Other extensions of the positive-stable distribution may also be interesting avenues for future research (e.g., polynomial tilting \citep{Devroye09}). However, our choice of \eqref{ExtendedPS} preserves the simplicity of the model while introducing a single parameter, the exponential tilting parameter $\theta$, that is directly connected to the dependence properties of the resulting $Z(\bs)$ process, while allowing for inference that is computationally tractable.
\begin{proposition}
Let $\{Z(\bs), \bs\in \mathcal{S}\}$ be defined as in \eqref{eq:z_process} with $A_1,\ldots,A_L\iid {\rm H}(\alpha,\alpha,\theta)$, $\alpha\in(0,1)$, $\theta\geq0$. Then, $Z(\bs)$ is max-id.
\end{proposition}
\begin{proof}
From \eqref{LaplaceExtendedPS}, the finite-dimensional distributions for $\{Z(\bs), \bs \in \mathcal{S}\}$ based on \eqref{ExtendedPS} are 
\begin{align}
\pr\{&Z(\bs_1)\leq z_1,\ldots,Z(\bs_D)\leq z_D\}=\pr\{\varepsilon(\bs_1)Y(\bs_1)\leq z_1,\ldots,\varepsilon(\bs_D)Y(\bs_D)\leq z_D\}\nonumber\\
&=\E\left(\pr\left[\varepsilon(\bs_1)\leq z_1\left\{\sum_{l=1}^L A_l K_l(\bs_1)^{1/\alpha}\right\}^{-\alpha},\ldots, \varepsilon(\bs_D)\leq z_D\left\{\sum_{l=1}^L A_l K_l(\bs_D)^{1/\alpha}\right\}^{-\alpha}\mid A_1, \ldots, A_L\right]\right)\nonumber\\
&=\E\left(\exp\left[-\sum_{j=1}^D z_j^{-1/\alpha} \sum_{l=1}^L A_l K_l(\bs_j)^{1/\alpha}\right]\right)\nonumber\\
&=\prod_{l=1}^L\E\left(\exp\left[-A_l \sum_{j=1}^D \left\{z_j/K_l(\bs_j)\right\}^{-1/\alpha} \right]\right)\nonumber\\
&=\exp\left(L\theta^\alpha-\sum_{l=1}^L\left[\theta + \sum_{j=1}^D \left\{z_j/K_l(\bs_j)\right\}^{-1/\alpha} \right]^\alpha\right).\label{eq:jointdistrMaxIDReichShaby}
\end{align}
As
\begin{equation*}
\pr\{Z(\bs_1)\leq z_1,\ldots,Z(\bs_D)\leq z_D\}^{1/n}=\exp\left\{L\left({\theta\over n^{1/\alpha}}\right)^\alpha-\sum_{l=1}^L\left[\left({\theta\over n^{1/\alpha}}\right) + \sum_{j=1}^D \left\{nz_j/K_l(\bs_j)\right\}^{-1/\alpha} \right]^\alpha\right\},
\end{equation*}
the finite-dimensional distributions, denoted $G(z_1,\ldots,z_D;\alpha,\theta)$, from this new process satisfy\\ $G(z_1,\ldots,z_D;\alpha,\theta)^{1/n}=G(nz_1,\ldots,nz_D;\alpha,\theta/n^{1/\alpha})$ for all $n\in\Natural$, and thus the process is max-id. This also confirms that the process is max-stable if and only if $\theta=0$. 
\end{proof}
In Section 2.3, we specify spatial priors for the basis functions, so Proposition 2.1 should be interpreted conditional on the basis functions.

Marginal distributions are no longer unit Fr\'echet when $\theta>0$; they may be expressed as
\begin{equation}
G_{\bs}(z)=\pr\{Z(\bs)\leq z\} =\exp\left(L\theta^\alpha-\sum_{l=1}^L\left[\theta + \left\{z/K_l(\bs)\right\}^{-1/\alpha} \right]^\alpha\right),\quad z>0.
\label{eq:maxid_margin}
\end{equation}

Bayesian and likelihood-based inference may be performed similarly as before, so this process enjoys the same computational benefits as the \cite{Reich12} model, while having the traditional max-stable \cite{Reich12} process as a special case on the boundary of the parameter space (i.e., when $\theta=0$). Note that unlike the \cite{Reich12} model, here the marginal distributions depend on the dependence parameters $\alpha$ and $\theta$, however, this is not a problem for inference as we adopt a copula-based approach, in which we separate the treatment of the marginal distributions and the dependence structure. Marginal modeling is described in greater detail in Section \ref{sub:marginal_modeling}. Finally, a spectral representation for the proposed max-id model is described in the Supplementary Material, which makes a link with the max-id models of \cite{Huser18}.

\subsection{Prior Specification for the Spatial Kernels Based on Flexible Log-Gaussian Process Factors} 
\label{sub:flexible_gaussian_process_factor_model}
The basis functions used in \citet{Reich12}, constructed from Gaussian densities, are radial functions, decaying symmetrically from their knot centers. While it is possible to approximate a wide range of extremal functions by considering a large collection of Gaussian density basis functions $K_1(\bs), \ldots, K_L(\bs)$ as in \eqref{eq:residual_dep_proc}, the resulting process is overly smooth and artificially non-stationary for fixed $L$. In this section, we propose an alternative prior for the basis functions, which allows for a parsimonious, yet flexible, stationary representation that can give insights into the predominant modes of spatial variability among of the underlying process.

More precisely, we extend the \cite{Reich12} model by replacing the Dirac prior on the Gaussian density basis functions with flexible log-Gaussian process priors, which more closely approximate the features of natural phenomena than radial basis functions. This choice of basis functions is analogous to the construction of the Brown-Resnick process \citep{Brown77, Kabluchko09}, which itself can be represented as the pointwise maximum over an infinite collection of scaled log-Gaussian processes. Let $\tilde{K}_l(\bs)$, $l = 1, \ldots, L-1$, be i.i.d.\ mean-zero stationary Gaussian processes, each with exponential covariance function, $C(h) = \delta^2\exp(-h/\rho), \, h \geq 0$, whose variance and range are $\delta_K^2 >0$ and $\rho_K>0$, respectively. We take the $L$th basis to be the constant function equal to the mean of the Gaussian process, i.e., $\tilde{K}_L(\bs) = 0$ for all $\bs\in\calS$. Fixing the $L$th term ensures that it is possible to recover the $\tilde{K}_l$ from the $K_l(\bs)$ terms, which is necessary for making posterior draws of $\tilde{K}_l(\bs)$ (see Supplementary Material). Other prior choices for the basis functions that may also be worth exploring include using a more general Mat\'{e}rn class of covariance functions or Gaussian processes with stationary increments and an unbounded variogram (i.e., fractional Brownian motions), akin to the Brown-Resnick process. Application of a fractional Brownian motion prior in this context would require a choice of origin for each basis function, which would increase the computational cost if one wanted to marginalize over that unknown origin, and so we do not pursue it here. To satisfy the sum-to-one constraint for each spatial location $\bs \in \calS$, we set
\begin{equation}
K_l(\bs) = \exp\left\{\tilde{K}_l(\bs)\right\}/\sum_{l = 1}^L \exp\left\{\tilde{K}_l(\bs)\right\}, \quad l = 1, \ldots, L.
\label{eq:ln_basis}
\end{equation}

The variance parameter $\delta_K^2$ controls the long-range spatial dependence of the max-id process $Z(\bs)$, with smaller values corresponding to stronger long-range dependence (see \cite{Davison12} for a similar discussion of geometric Gaussian processes). When $\delta_K^2$ is large, the difference in relative magnitudes of the unnormalized log-Gaussian processes at any given location $\bs$ is likely to be larger than when $\delta_K^2$ is small. Normalizing the basis functions when the difference in magnitudes is great gives way to more volatile fluctuations between dominating basis functions, and hence less long-range dependence. The Gaussian process range parameter $\rho_K$ governs the short-range dependence, now with larger values corresponding to stronger short-range dependence. Because the proposed basis functions provide greater flexibility in adapting to the data than the fixed Gaussian density basis, fewer basis functions are needed. In the data application presented in Section \ref{sec:application_to_extreme_precipitation}, we choose the number of basis functions using an out-of-sample log-score criterion. Increasing the number of basis functions allows for greater flexibility in capturing spatially dependent subregions that tend to have extreme events together at the cost of greater computational burden.

When the deterministic basis functions used by \cite{Reich12} are replaced with random ones, the max-stability (when $\theta = 0$) and max-infinite divisibility properties should be interpreted conditionally on the basis functions. Both the conditional and unconditional dependence properties are described in Section \ref{sub:dependence_properties}.

\subsection{Marginal Modeling and Realizations} 
\label{sub:marginal_modeling}
For marginal distribution modeling, we use the Generalized Extreme-Value (GEV) distribution, which is the asymptotic distribution for univariate block maxima. The ${\rm GEV}(\mu, \sigma, \xi)$ distribution function has the following form:
\begin{equation*}
G(z) = \begin{cases}
	\exp\left[-\exp\left\{-(z-\mu)/\sigma\right\}\right], &\xi = 0, \\
	\exp[-\{1 + \xi (z - \mu)/\sigma\}_+^{-1/\xi}], & \xi \neq 0,
\end{cases}
\end{equation*}
where $a_+=\max(0,a)$, for some location $\mu\in \mathbb{R}$, scale $\sigma > 0$, and shape $\xi \in \mathbb{R}$ parameters, with support $\{z\in\mathbb R: 1 + \xi (z-\mu)/\sigma) > 0\}$ when $\xi \neq 0$, and $\mathbb{R}$ when $\xi = 0$. Since monotone increasing transformations of the marginal distributions do not change the max-id or max-stable dependence structure, we allow for general GEV marginal distributions that are possibly different for each spatial location. In other words, we set $\tilde{Z}(\bs) = \mathrm{GEV}^{-1}[G_{\bs}\{Z(\bs)\}; \mu(\bs), \sigma(\bs), \xi(\bs)]$, where $G_\bs(z)$ is the marginal distribution of $Z(\bs)$, which in the case of the \cite{Reich12} model is $G_\bs(z) = \exp \left(-z^{-1}\right), z > 0$, and in the $\theta>0$ case is given in \eqref{eq:maxid_margin}, and $\mathrm{GEV}^{-1}\{\cdot;\mu(\bs), \sigma(\bs), \xi(\bs)\}$ is the quantile function for a GEV distribution with location $\mu(\bs)$, scale $\sigma(\bs) > 0$, and shape $\xi(\bs)$. We treat $\tilde{Z}(\bs)$ as our response. In subsequent sections, Gaussian process priors are assumed for the GEV parameters $\mu(\bs)$, $\gamma(\bs) = \log\left\{\sigma(\bs)\right\}$, and $\xi(\bs)$, and Markov chain Monte Carlo (MCMC) methods are used to draw posterior samples for this model. The details of the MCMC sampler are given in the Supplementary Material.

To visualize some of the features of our model, we present some sample paths in Figure \ref{fig:realizations}. Realizations of $\tilde{Z}(\bs)$ on the unit square constructed using the Gaussian density ($L = 25$ evenly spaced basis functions, with standard deviation $\tau = 1/6$) and log-Gaussian process ($L = 15$ basis functions, with variance $\delta_K^2 = 25$ and range $\rho_K = 3/4$) basis functions are shown in Figure \ref{fig:realizations}. For illustration, the realizations have standard Gumbel margins everywhere in space, i.e., $\mu(\bs) = \xi(\bs) = 0$ and $\sigma(\bs) = 1$ for all $\bs\in \mathcal{S}$. The figure illustrates the role of $\alpha$ in controlling the relative contribution of the nugget process, and the impact of $\theta$ on the asymptotic dependence structure. Weaker tail dependence is present in the max-id models ($\theta>0$) than their max-stable counterparts ($\theta=0$). Moreover, the general shapes of the Gaussian density basis model realizations appear less resemblant of natural processes than do those from the log-Gaussian process basis model.

While we have only developed the model for a single realization of the process $\tilde{Z}(\bs)$ so far, the model can easily be generalized to accommodate multiple replicates in time, which we will use in Section \ref{sec:application_to_extreme_precipitation}. In particular, treating time replicates of the process to be independent, we denote the maxima process observed at spatial location $\bs$ and time $t$ by $\tilde{Z}_t(\bs)$, $t = 1, \ldots, T$.  We assume the marginal GEV parameters and basis functions do not vary in time, but allow the relative contribution of each basis function to be different for different time replicates of the process by taking the random basis coefficients to be $A_{l,t} \iid {\rm H}(\alpha,\alpha,\theta)$, $l = 1,\ldots, L$, and $t = 1,\ldots, T$.

\begin{figure}[t!]
  \begin{center}
		\includegraphics[width=\textwidth]{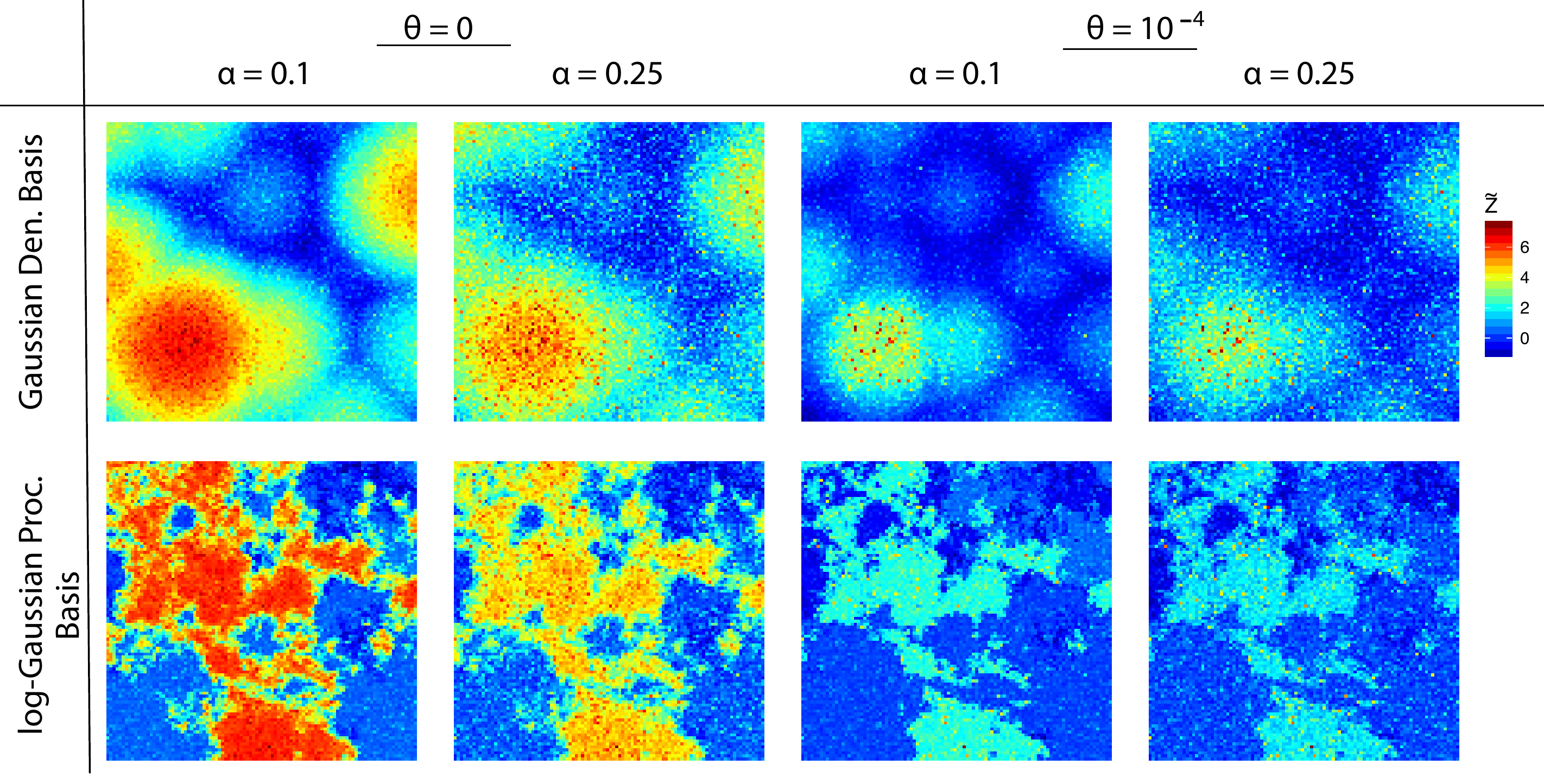}
		\caption{Realizations of the max-stable ($\theta=0$) and max-id ($\theta>0$) processes with Gaussian density (top) and log-Gaussian process (bottom) basis functions, plotted on Gumbel margins.}
		\label{fig:realizations}
	\end{center}
\end{figure}

\subsection{Dependence Properties} 
\label{sub:dependence_properties}
In this section, we explore the dependence properties of the proposed max-id model. The parameter $\theta$ plays a crucial role in determining the asymptotic dependence class. \cite{Reich12} show that $\{Z(\bs), \bs \in \mathcal{S}\}$ is asymptotically dependent and max-stable for $\alpha \in (0,1)$, $\theta=0$. However, when $\theta > 0$, this is no longer the case.
\begin{proposition}
The process $\{Z(\bs), \bs \in \mathcal{S}\}$ defined in Sections \ref{sub:subasymptotic_modeling_maxid}--\ref{sub:flexible_gaussian_process_factor_model} using the log-Gaussian process basis prior in \eqref{eq:ln_basis} is an asymptotically independent process when $\theta > 0$ and asymptotically dependent when $\theta = 0$ and $\alpha < 1$.
\end{proposition}
For a proof, see Appendix \ref{sec:tail_dep_properties}. Figure \ref{fig:chi_plots} displays two common dependence measures, $\chi_u = \pr\left[G_{\bs_1}\{Z(\bs_1)\}>u\mid G_{\bs_2}\{Z(\bs_2)\}>u\right]$ and $\bar{\chi}_u = {2 \log \pr\left[G_{\bs_2}\{Z(\bs_2)\} > u\right]\over\log \pr\left[G_{\bs_1}\{Z(\bs_1)\} > u, G_{\bs_2}\{Z(\bs_2)\} > u\right]}$, $0<u<1$ \citep{Coles99} to illustrate the role of $\alpha$ and $\theta$ in controlling the dependence properties of the tail process. Although notationally we have omitted the dependence of $\chi_u$ on $\bs_1$ and $\bs_2$, $\chi_u$ will also depend on the locations in the (non-stationary) Gaussian density basis case. Nevertheless, while the \cite{Reich12} max-stable process is non-stationary, it is approximately stationary for a dense set of spatial knots. An attractive feature of the proposed model is that as $\theta \downarrow 0$, $\chi_u$ and $\bar{\chi}_u$ transition smoothly from weak dependence to strong dependence for all $u < 1$.

\begin{figure}[t!]
  \begin{center}
		\includegraphics[width=1\textwidth]{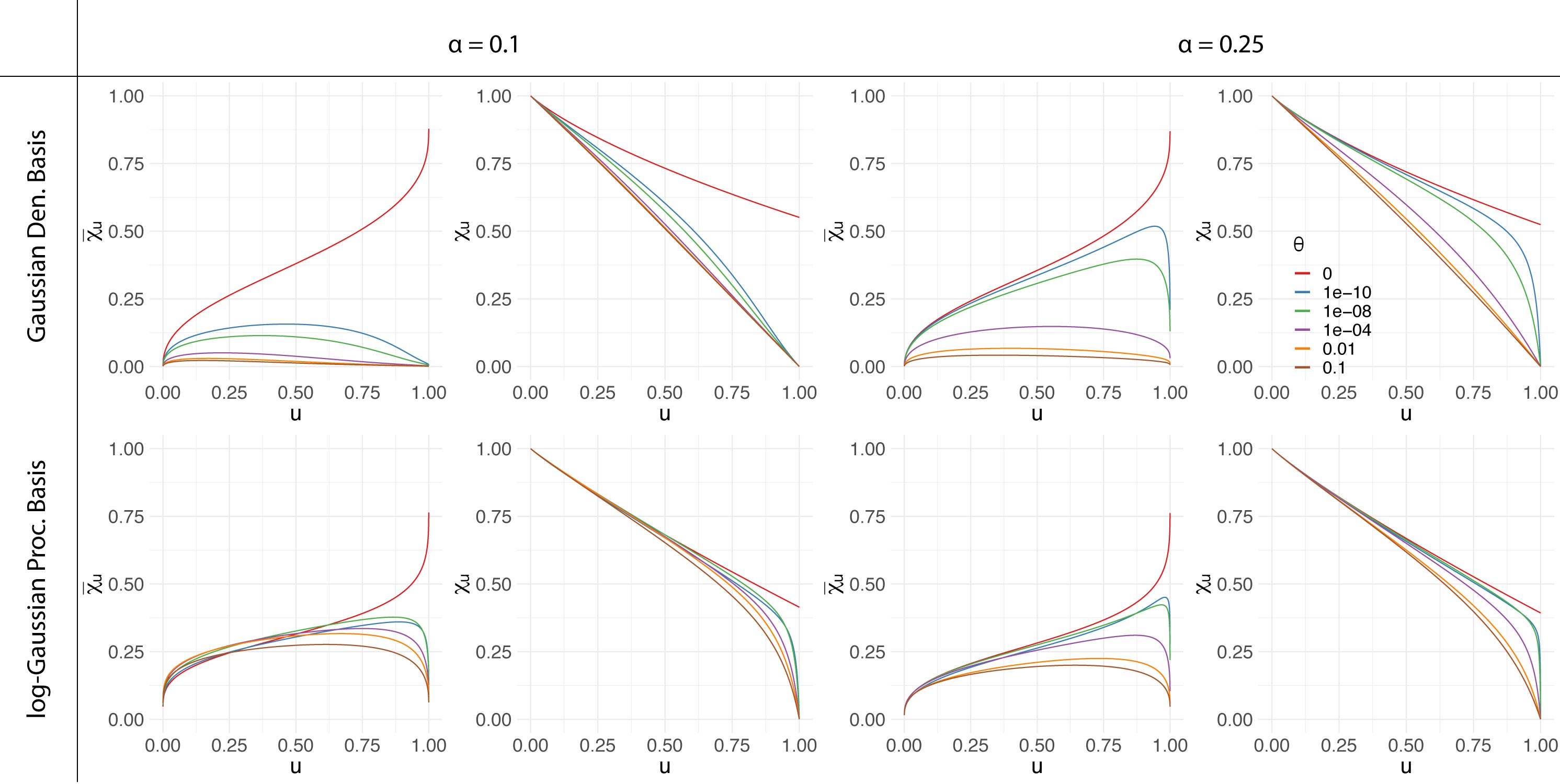}
		\caption{Dependence measures $\bar{\chi}_u$ and $\chi_u$ for the max-stable ($\theta = 0$) and max-id ($\theta > 0$) models for $Z(s), s \in \mathbb{R}$, using $L = 25$ Gaussian density ($\tau = 1/6$) and $L = 15$ log-Gaussian process ($\delta^2_K = 25$, $\rho_K = 3/4$) basis functions for $s_1 = 0$ and $s_2 = 1/4$. The knots of the Gaussian density basis functions are evenly spaced between $0$ and $1$. The figures in the bottom row correspond to $\chi_u$ after marginalizing over the log-Gaussian process basis functions based on $M = 1,000$ Monte Carlo draws.}
		\label{fig:chi_plots}
	\end{center}
\end{figure}
The extremal coefficient $\theta_{\mathcal{D}}$, studied by \citet{Schlather03}, is a measure of spatial dependence along the diagonal of the finite-dimensional distributions of max-stable processes. It takes on values from $\theta_{\mathcal{D}} = 1$ when the components are perfectly dependent to $\theta_{\mathcal{D}} = D$ when they are independent, and therefore can be interpreted as the effective number of independent variables. The finite-dimensional distributions of a max-stable process with unit-Fr\'{e}chet margins at level $z$ can be written in the form
\begin{equation}
\Pr\left\{Z(\bs_1)\leq z, \ldots, Z(\bs_D) \leq z\right\} = \exp\left\{-{\theta_{\mathcal{D}}(\bs_1, \ldots, \bs_D)\over z}\right\}, \quad \theta_{\cal{D}}(\bs_1, \ldots, \bs_D) \in [1,D],
\label{eq:ext_coef}	
\end{equation}
where $\theta_{\cal{D}}$ determines the spatial dependence and does \emph{not} depend on the level $z$.  The rigidity of the dependence structure across all quantiles limits the applicability of max-stable models to processes that exhibit varying spatial dependence types at different quantiles. From \eqref{eq:jointdistrMaxIDReichShaby}, we can see that the max-id extension of the \cite{Reich12} model does not possess this property for $\theta>0$.

Figure \ref{fig:spatial_chi} contrasts the spatial dependence features of the proposed models. We examine how the conditional probability of jointly exceeding a fixed quantile decays with increasing distance. Each panel shows the spatial decay of $\chi_u$ as a function of increasing spatial lag $h$ for several quantiles. We see qualitatively different behavior in the spatial decay of dependence at different quantiles between the max-stable and max-id models. In the max-stable cases, the conditional exceedance probability $\chi_u$ at short spatial lags $h$ is very similar at all levels $u$ of the distribution. The max-id models allow for more flexibility, as can be seen by the attenuated curves for higher quantiles and wider array of spatial decay types. From Figure \ref{fig:spatial_chi}, it can be seen that for $\theta >0$, the parameter $\alpha$ plays a role in how precipitous the decay in spatial dependence is with increasing distance, with smaller $\alpha$ corresponding to steeper decay. Also, just as in \cite{Reich12}, $\alpha$ determines the contribution of the nugget effect, which is greater when $\alpha$ is large and lesser when $\alpha$ is small.

\begin{figure}[t!]
  \begin{center}
		\includegraphics[width=0.75\textwidth]{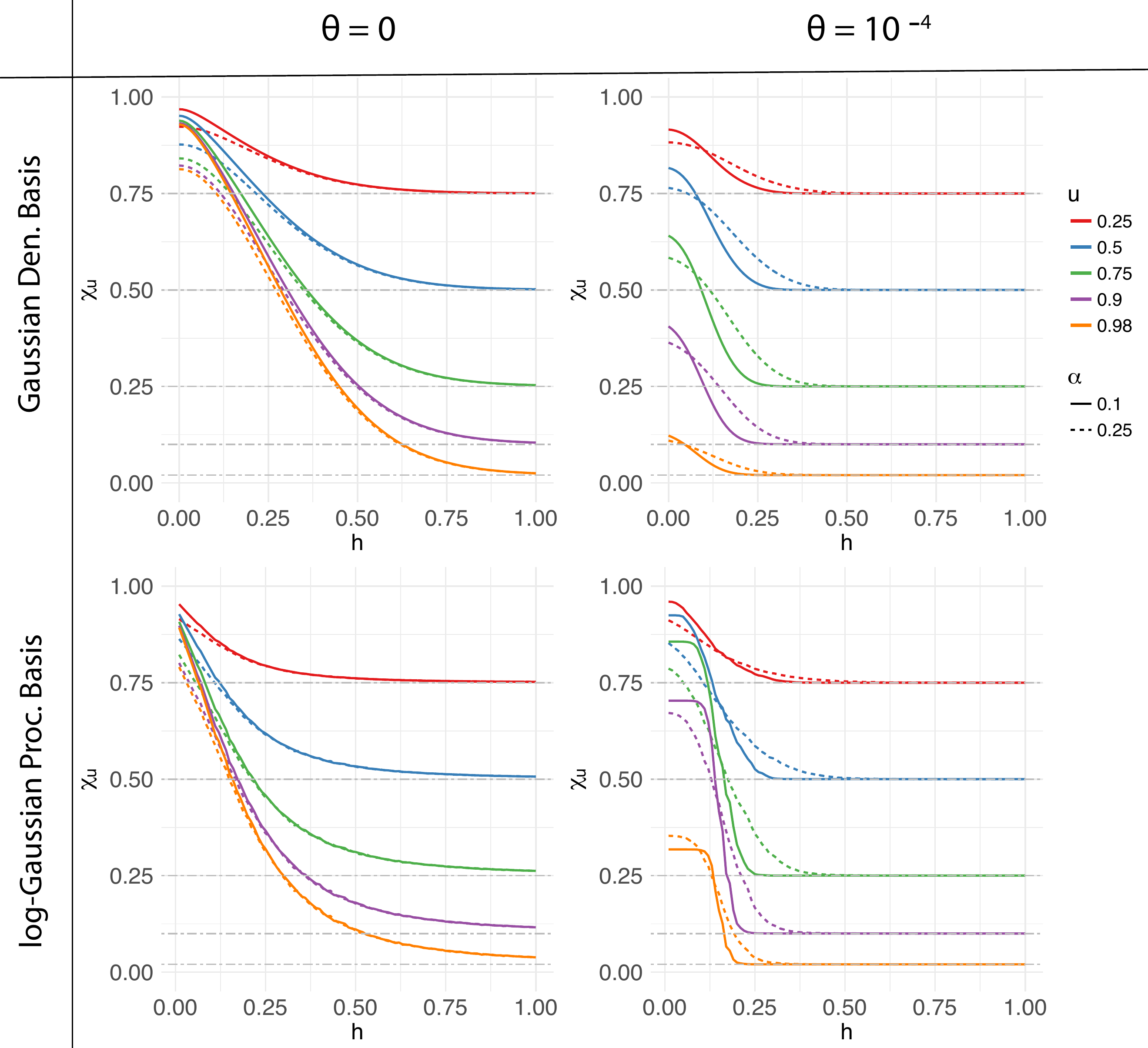}
		\caption{Dependence measure $\chi_u(h)$ between $Z(s_0)$ and $Z(s_0 + h)$ for $s_0 = 0$ as a function of $h$ for max-stable (left column) and max-id (right column) models on $\calS=[0,1]$, with $L = 25$ Gaussian density basis functions with $\tau = 1/6$ (top row) and $L = 15$ log-Gaussian process basis functions with  $\delta_K^2 = 25$ and $\rho_K = 3/4$ (bottom row) basis functions for varying $\alpha$ and $u$. Gaussian density basis functions are evenly spaced between 0 and 1. Estimates of $\chi_u(h)$ in the log-Gaussian process basis model are based on 50,000 Monte Carlo replicates. Horizontal dash-dot gray lines representing the values of $\chi_u$ for independent $Z(s_0)$ and $Z(s_0 + h)$ are plotted for reference.}
		\label{fig:spatial_chi}
	\end{center}
\end{figure}

To confirm that our MCMC algorithm produces reliable results, and to evaluate the algorithm's ability to infer the parameters under different regimes, we conduct a simulation study for both the Gaussian density basis and the log-Gaussian process basis models. The simulation study design and results are described in detail in the Supplementary Material. In all scenarios considered, credible intervals achieve nearly nominal levels, confirming the reliability of our MCMC algorithm.

\section{APPLICATION TO EXTREME PRECIPITATION}
\label{sec:application_to_extreme_precipitation}

\subsection{Data and Motivation} 
\label{sub:data_and_motivation}
In this section, we apply our model to extreme precipitation over the northeastern United States and Canada. Our aim is to understand the spatial dependence of extreme precipitation while accounting for measurement uncertainty. The data for this application were obtained from \url{https://hdsc.nws.noaa.gov/hdsc/pfds/pfds_series.html}, which is maintained by the National Oceanic and Atmospheric Administration (NOAA). Observations consist of annual maximum daily precipitation accumulations (in inches) observed between 1960 and 2015 at $N = 646$ gauge stations (see Figure \ref{fig:gauge_knot_locations}). The observation at gauge location $\bs_i$, $i= 1,\ldots, 646$, and year $t = 1, \ldots, 56$, is denoted by $\tilde{Z}_t(\bs_i)$.

\begin{figure}[t!]
  \begin{center}
		\includegraphics[width=0.8\textwidth]{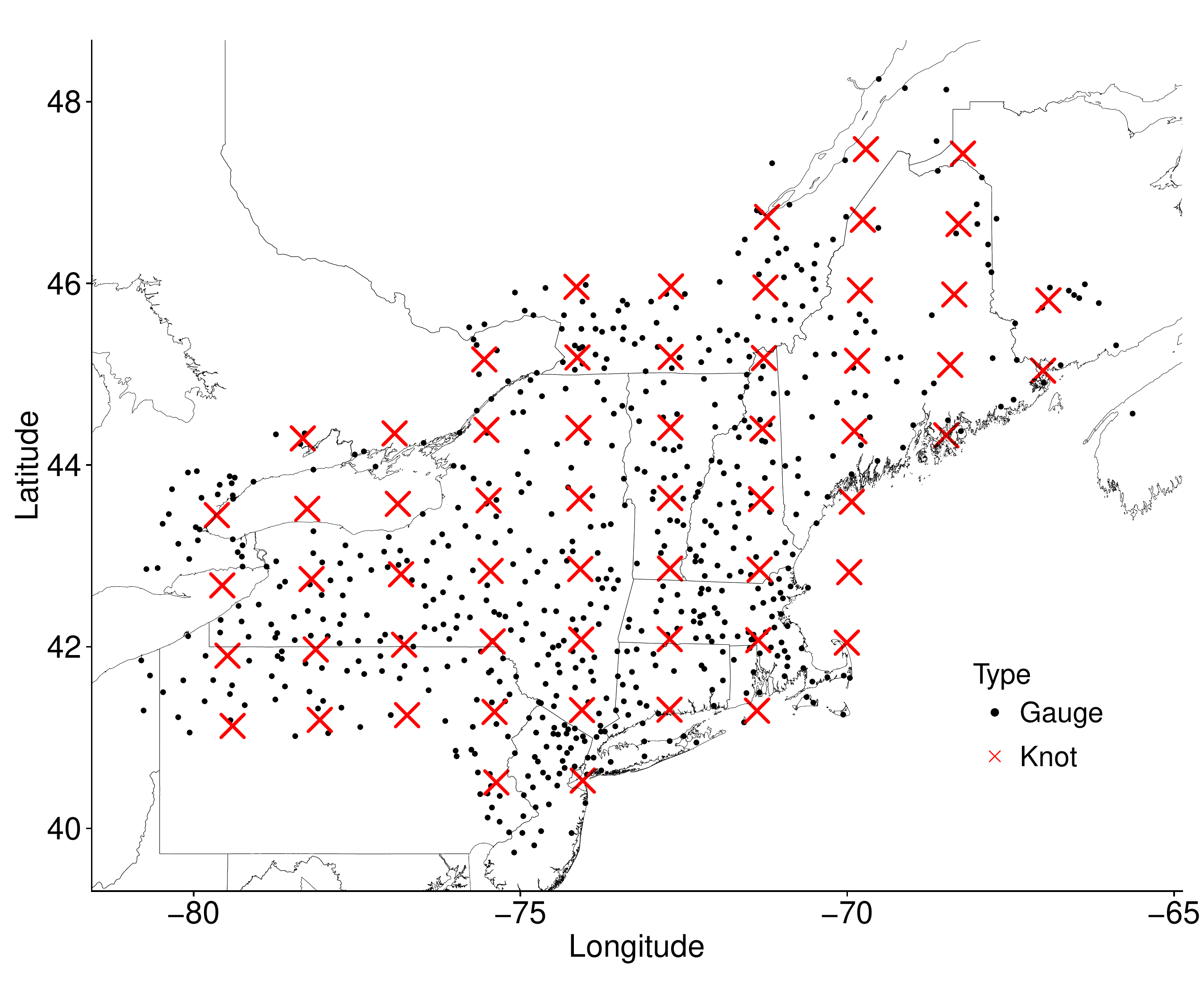}
		\caption{Precipitation gauge locations ($N = 646$) across the northeastern United States and Canada are plotted as black dots and Gaussian density basis knot locations ($L = 60$) are plotted as red crosses.}
		\label{fig:gauge_knot_locations}
	\end{center}
\end{figure}

\subsection{Model Fitting and Validation} 
\label{sub:model_estimation_and_validation}
The precipitation data are analyzed by applying the four max-id models described in Section \ref{sec:hierarchical_construction_of_spatial_maxid}, namely (M1) Gaussian density basis, $\theta = 0$; (M2) Gaussian density basis, $\theta > 0$; (M3) log-Gaussian process basis, $\theta = 0$; and (M4) log-Gaussian process basis, $\theta > 0$, where realizations of the process for each year are treated as i.i.d.\ replicates. Although further temporal dependence and trends could be modeled in both the GEV marginal parameters and basis scaling factors $A_{l,t}$, Kwiatkowski-Phillips-Schmidt-Shin tests \citep{Kwiatkowski92} for temporal non-stationarity among the annual maxima were performed separately for each station, and 85\% of stations yielded no evidence for temporal non-stationarity at confidence level $95\%$. The proposed model would be more complex and computationally demanding to fit if one were to account for temporal non-stationarity. Therefore, for the sake of simplicity, and since overall the data do not appear to be highly non-stationary over time, we will ignore this aspect in our analysis. Accounting for temporal non-stationarity would be an interesting avenue of future research to further develop this model.

In particular, both the dependence model and GEV marginal distributions are assumed to be constant over time.  We assume independent Gaussian process priors, each with constant mean $\beta_\psi \sim N(0, 100)$ and stationary exponential covariance function $C(h) = \delta^2_\psi \exp(-h/\rho_\psi), h \geq 0$, $\psi\in\{\mu, \gamma\}$, on the location $\mu(\bs)$ and log-scale $\gamma(\bs) \equiv \log\{\sigma(\bs)\}$ marginal parameters of the GEV distribution, with half-normal priors for $\delta^2_\psi \sim N_+(0, 100)$ and $\rho_\psi \sim N_+(0, \max_{i,j}(||\bs_i - \bs_j||)^2)$. Due to the difficulty in estimating the shape parameter \citep{Cooley07,Opitz18}, we use a spatially constant prior, $\xi \sim N(0, 100)$. The dependence parameter priors are as follows: For $\alpha$ and $\theta$, we take $\alpha \sim {\rm Unif}(0,1)$ and $\theta \sim N_+(0,100)$. For the Gaussian density basis models, we use $L = 60$ knot locations on an evenly spaced grid (see Figure \ref{fig:gauge_knot_locations}). A half normal prior is put on the Gaussian density bandwidth parameter $\tau \sim N_+(0, \max_{i,j}(||\bs_i - \bs_j||)^2)$. In the case of the log-Gaussian process basis models, we consider $L = 10, 15,$ and $20$ basis functions. More basis functions enable better representation of the data, but at the risk of overfitting. Priors $\delta_K^2 \sim N_+(0, 100)$ and $\rho_K \sim N_+(0, \max_{i,j}(||\bs_i - \bs_j||)^2)$ are assumed for the exponential covariance parameters. Handling missing values is straightforward using the proposed approach. For each iteration of the MCMC algorithm, missing values are sampled from the posterior predictive distribution; this is detailed in the Supplementary Material. We run each MCMC chain under two different parameter initializations for 40,000 iterations using a burn-in of 10,000 with data from $546$ stations, reserving $100$ stations for model evaluation. Some of the parameters, particularly $\beta_\mu$ and $\beta_\gamma$, were quite slow to converge. In all four cases, the posterior densities were similar across the two initializations. 

It is currently not possible to fit existing max-stable, inverted-max-stable \citep{Wadsworth12}, and other max-id models \citetext{see, e.g. \citealp{Huser18}, \citealp{Padoan13}} using a full likelihood or Bayesian approach when the number of spatial locations is large; see \cite{Castruccio16}, \cite{Dombry17} and \cite{Huser19}. Under these constraints, a natural alternative for comparison is the model for block maxima proposed by \cite{Sang10}, which also belongs to the asymptotic independence class. Specifically, let $\{W(\bs), \bs \in \mathcal{D}\}$ be a mean-zero Gaussian process with exponential correlation function and unit variance. The annual maxima are then modeled as $Z(\bs) = \mathrm{GEV}^{-1}[\Phi\{W(\bs)\}; \mu(\bs), \sigma(\bs), \xi]$, where the location $\mu(\bs)$ and $\log\{\sigma(\bs)\}$ each follow mean zero Gaussian processes with exponential covariance functions, with the same priors as above, and $\Phi$ denotes the standard normal distribution function. We refer to this as the the GEV-Gaussian process copula model.

To compare models, we calculate out-of-sample log-scores \citep{Gneiting07}, for annual maxima at the 100 holdout stations, which is simply the log-likelihood of the holdout data for each model based on conditional predictive simulations of the latent model parameters at the unobserved sites. Since the log-scores are calculated on holdout data, they implicitly account for model complexity. We also emphasize that because the predictions are based on the joint likelihood, the log-scores reflect not only the marginal fits, but also how well the model captures the dependence characteristics of the observed data. The best log-score (higher scores are better) of the two initializations for each model is reported in Table \ref{tab:log_scores}. The max-id models ($\theta>0$) outperform their max-stable counterparts ($\theta=0$). The log-score for the GEV-Gaussian process copula model is worse than the other models considered. The estimated marginal surfaces are similar across all of the models considered, indicating that the misspecification is due to differences in the dependence model for the annual maxima.

The max-id, log-Gaussian process basis model with $\theta >0$ and $L = 15$ basis functions has the highest log-score (shown in bold), suggesting it should be preferred among the considered models for this data application, and as such we focus on this model for the remainder of our analysis. For this model, the posterior mean (95\% credible interval) estimates of the dependence parameters are $0.725 \,\,(0.702, 0.747)$ for $\alpha$, $0.024 \,\,(0.006, 0.060)$ for $\theta$, and for the spatial basis functions $33.9 \,\,(23.8, 47.2)$ for $\delta_K^2$ and $462 \,\,(332, 642)$ miles for $\rho_K$, suggesting the presence of some residual dependence beyond that explained by spatially-varying marginal parameters. Also, while we have specified vague priors on the model parameters, the posterior distributions are highly concentrated around their corresponding posterior means. Although the proposed inference scheme does not allow for jumps between $\theta = 0$ and $\theta > 0$, the posterior samples of $\theta$ are still somewhat informative about the asymptotic dependence class. In particular, since the dependence properties of our model are smooth in $\theta$ at zero, the fact that the 95\% credible interval for $\theta$ is relatively symmetric and distant from 0 gives support for asymptotic independence among precipitation extremes.

To validate the decision of having the same dependence parameters $\alpha$ and $\theta$ over the entire region, log-Gaussian process basis models with $\theta > 0$ were also separately fitted to four subregions, two inland and two coastal. The 95\% credible intervals for $\alpha$ and $\theta$ overlap with those fitted to the entire region, suggesting homogeneous spatial dependence of the process over the study region.

\begin{table}[t!]
\begin{center}
\centering
\caption{Log-scores estimated from annual maxima observed at the holdout stations are used to compare the four models presented in Section \ref{sec:hierarchical_construction_of_spatial_maxid}, and the GEV-Gaussian process copula model. Higher log-scores correspond to better fit. The max-id, log-Gaussian process basis model has the highest log-score (shown in bold).}
\label{tab:log_scores}
\begin{tabular}{r|c|ccc|c}
\multicolumn{1}{c}{} & \multicolumn{1}{c}{Gaussian Density Basis} &  \multicolumn{3}{c}{log-Gaussian Process Basis}  & GEV-Gaussian Process Copula\tabularnewline
L  & 60 & 10 & 15 & 20 & \\
\hline
$\theta = 0$ &-5292.5&  -5410.7 &-5406.4& -5415.2 & \multirow{2}{*}{-6097.048}\\
$\theta >0$ &-5218.3 & -5194.6 &\textbf{-5172.6} & -5207.9 \\
\end{tabular}
\end{center}
\end{table}

Further, to examine the model fit, we compare empirical and model-based estimates of $\chi_u$ as a function of spatial lag $h$ and threshold $u$ for the holdout stations (Figure \ref{fig:lnmid_precip_chi}). 
\begin{figure}[t!]
  \begin{center}
		\includegraphics[width=\textwidth]{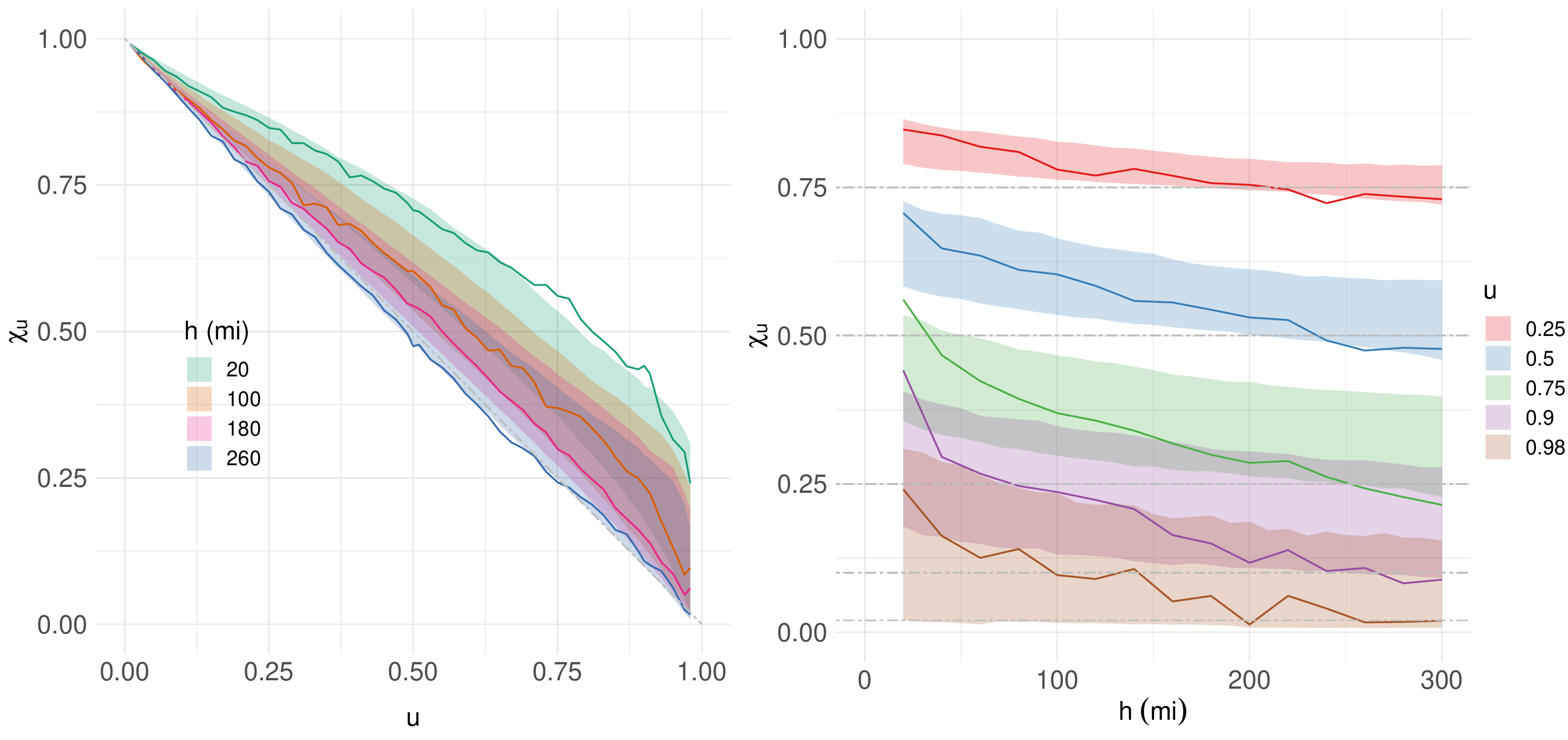}
		\caption{The left panel shows $\chi_u$ as a function of $u$ for fixed spatial lags $h = 20, 100, 180, 260$ miles calculated for the 100-holdout stations. Empirical estimates are shown as a solid black line, and max-id, log-Gaussian process basis model 95\% credible intervals are shown as gray ribbons. The decay of $\chi_u$ towards zero as $u \rightarrow 1$ suggests that daily precipitation are asymptotically independent. To understand the spatial dependence of extreme precipitation at increasingly extreme levels, empirical (solid lines) and model 95\% credible intervals (ribbons) of $\chi_u(h)$ for the holdout stations are plotted for several quantiles $u = 0.25, 0.5, 0.75, 0.9, 0.98$ (right panel). Horizontal dash-dot gray lines representing the values of $\chi_u$ under an everywhere-independent model are plotted for reference. The plot shows good overall agreement between the model fits and empirical estimates.}
		\label{fig:lnmid_precip_chi}
	\end{center}
\end{figure}
The left panel shows $\chi_u$ as a function of $u$ for at fixed lags $h=20,100,180,260$ miles, and the right panel shows the spatial decay of $\chi_u$ as a function of spatial lag $h$ for several fixed marginal quantiles $u = 0.25, 0.5, 0.75, 0.9, 0.98$. Empirical estimates are represented by solid lines and $95\%$ credible intervals for each model by shaded ribbons. From the left panel, we can see that the max-id model captures the asymptotic independence behavior of the precipitation data quite well. The max-stable model slightly underestimates the relatively strong dependence at shorter distances, but with comparable coverage to the max-id model at other distances (see Supplementary Material). The slight discrepancy at shorter distances may be due to the phenomenon described by \cite{Robins00} wherein intervals from posterior summaries like $\chi_u$ that are calculated from MCMC draws are too narrow. From the right panel, we deduce that the annual maximum precipitation data exhibit quite strong spatial dependence up to about $200$ miles, with weaker spatial dependence at higher quantiles. Moreover, $\chi_u$ decays towards its independence level as a function of distance $h$ faster at the $0.9$ and $0.98$ quantiles than at the $0.25$ and $0.5$ quantiles. 

In order to assess the joint spatial prediction skill of our model, we display in Figure \ref{fig:lnmid_qqplots} quantile-quantile (QQ)-plots for group-wise summaries of the annual maxima taken over the $100$ holdout stations (see \citet{Davison12} for a similar analysis). The results show adequate correspondence between the model-based and empirical quantiles of the group-wise means, whereas the observed group-wise minima (maxima, respectively) appear to be slightly underestimated (overestimated, respectively) by the model. Corresponding QQ-plots when $\theta = 0$ (not shown) give similar patterns with minima (maxima, respectively) lying slightly further above (below, respectively) the 95\% credible intervals.

\begin{figure}[t!]
  \begin{center}
		\includegraphics[width=\textwidth]{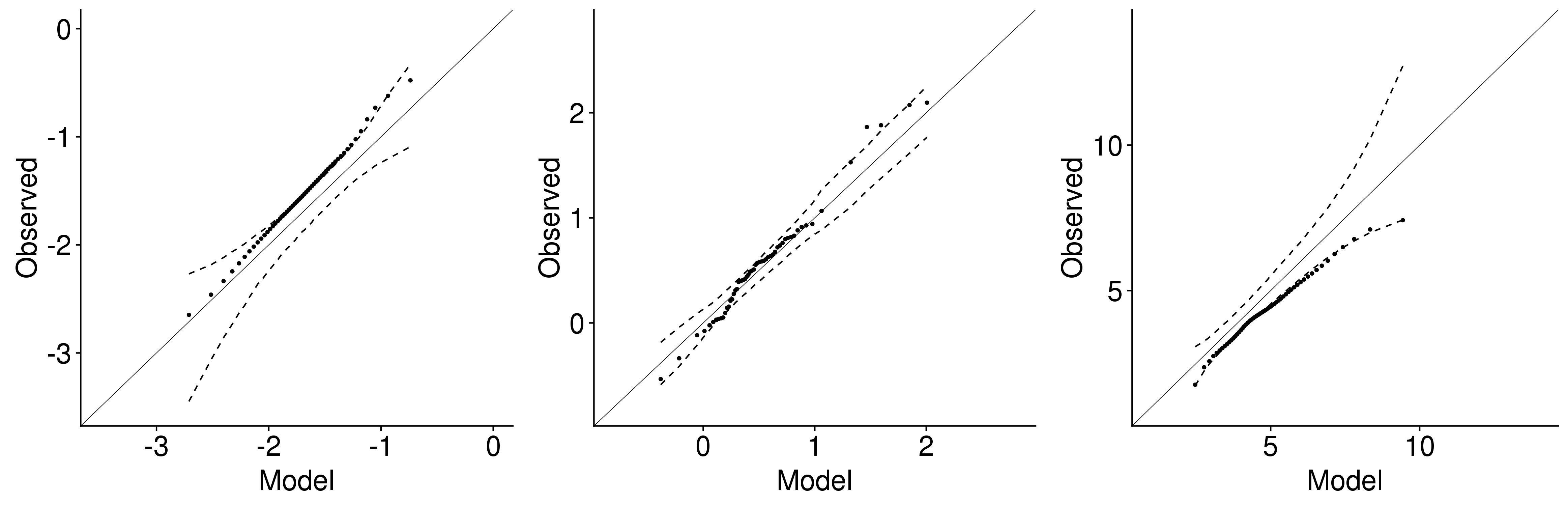}
		\caption{QQ-plots of the observed and predicted group-wise minima (left), mean (center), and maxima (right) taken over the annual maxima from all 100 holdout stations. The dashed lines represent 95\% credible intervals. The plots reflect reasonable correspondence between the empirical and modeled multivariate distributions. To account for the fact that the marginal GEV distributions vary across stations, observations are first transformed to unit Gumbel scale using the probability integral transform for the GEV marginal distributions at each station from the fitted model.}
     	\label{fig:lnmid_qqplots}
	\end{center}
\end{figure}

Maps of the marginal posterior predictive means and standard deviations of the 0.99 quantile of annual maxima (i.e., 100-year return level) for the max-id, log-Gaussian process basis model are shown in Figure \ref{fig:lnmidQ99}. The posterior mean surfaces are consistent with marginal quantile surfaces for the region as reported in NOAA Atlas 14 \citep{Perica13}. The posterior standard deviation surface shows the greatest variability in Maine, Long Island, and along the boundary of the observation region where there are relatively few gauge locations. For illustration, observed maxima in 2012 and the posterior predictive mean for that year are plotted in Figure \ref{fig:ppred_2012}. Recall that only the scaling factors $A_{l,t}$ vary in time. The posterior predictive mean appears to capture the general spatial trend of the maxima observed in 2012 well. 
\begin{figure}[t!]
  \begin{center}
		\includegraphics[width=\textwidth]{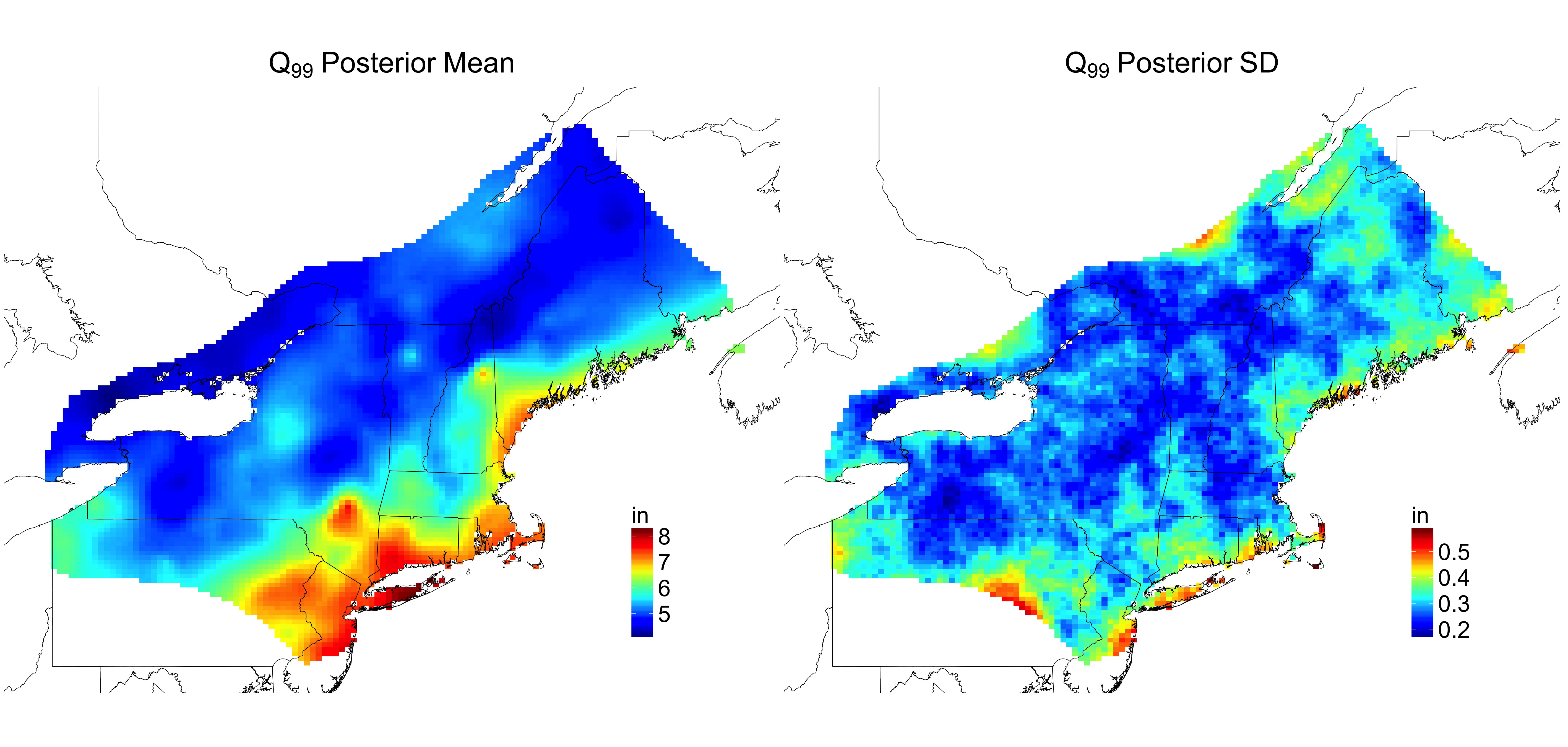}
		\caption{Pointwise posterior predictive mean (left) and standard deviation (right) of the 100-year return level of daily precipitation.}
		\label{fig:lnmidQ99}
	\end{center}
\end{figure}

\begin{figure}[t!]
  \begin{center}
		\includegraphics[width=\textwidth]{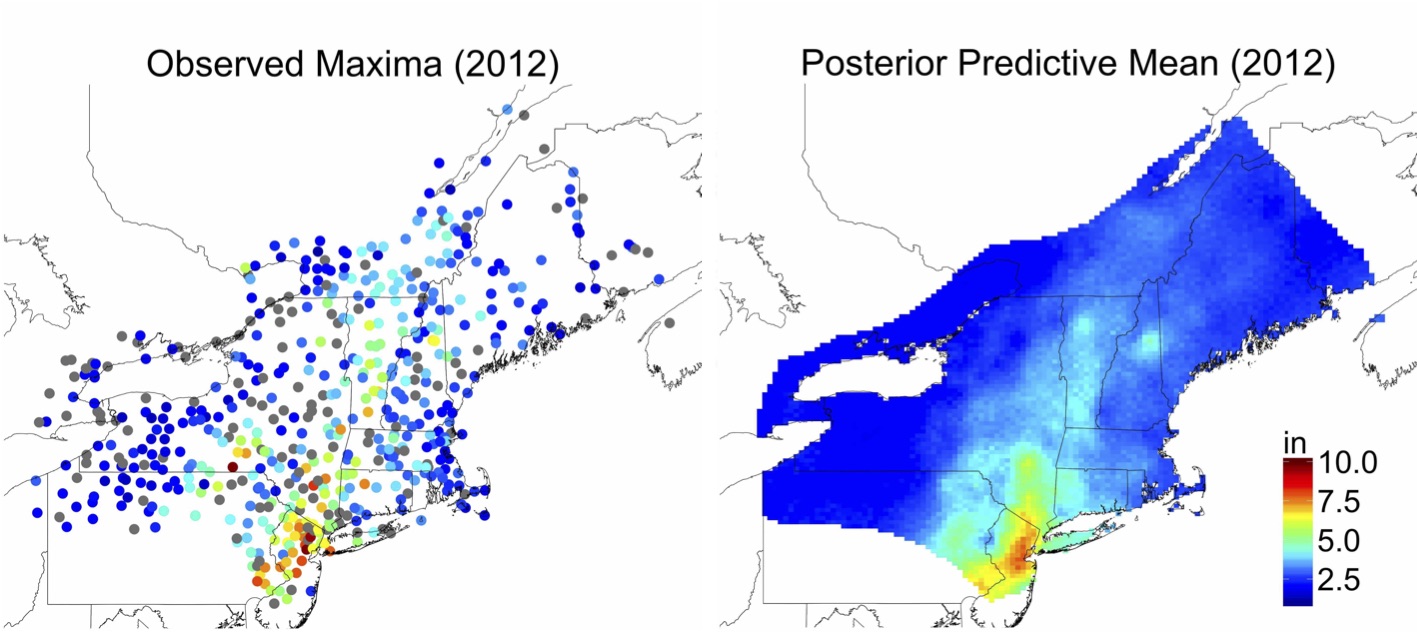}
		\caption{Observed precipitation accumulations (left), a single posterior predictive draw (middle), and posterior predictive means (right) for the year 2012. Missing values are shown in gray.}
		\label{fig:ppred_2012}
	\end{center}
\end{figure}

\subsection{Principal Modes of Spatial Variability Among Precipitation Extremes} 
\label{sub:analysis_of_the_principal_modes_of_spatial_variability_among_precipitation_extremes}
Spatial principal component analysis (PCA) \citep{Demsar13,Jolliffe02} and Empirical Orthogonal Functions \citep{Hannachi07} have proven to be useful methods for exploring the main large scale features of spatial processes. However, aside from recent work by \citet{Morris16} and \citet{Cooley18}, little has been done to this end for spatial extremes. The model we have proposed allows for an exploratory visualization that is very similar to a spatial PCA method that \citet{Demsar13} refers to as Atmospheric Science PCA in their review of Spatial PCA methods, where the data consist of time replicates of a univariate spatial process observed at several locations. 

An attractive feature of the log-Gaussian process basis model is that it provides a low-dimensional representation of the predominant modes of spatial variability among extremes. Analogously to factor analysis, the primary spatial trends among extreme precipitation can be described by a subset of the spatial basis functions $K_l(\bs)$ that contribute the most to the overall process. To achieve this, motivated by PCA factorization, which finds the directions of maximum variance in the data, we rank the spatial basis functions $K_l(\bs)$ $l = 1,\ldots, L$, by the posterior year-to-year variation of their corresponding basis coefficients $A_{l,t}$ (i.e., higher posterior variance corresponds to lower rank). Arguably, both the means and variances of the coefficients $A_{l,t}$ play a role in the relative contribution of the corresponding basis function to the overall process. However, from inspection, the basis coefficients with the highest posterior variance also have the highest posterior means. Examining the variance of the basis coefficients for each $l = 1, \ldots, L$, against their ranks give a rough indication of the number of basis functions with sizable contributions to the overall process. Also, while label switching is possible, from inspection of the MCMC samples of the basis functions, this does not appear to be a major concern for this application. If label switching is present, application of the pivotal reordering algorithm proposed by \cite{Marin05, Marin07} can be used to permute the labels of the basis functions and scaling factors before ranking the basis functions. Posterior means of the first six spatial basis functions are shown in Figure \ref{fig:top_K}. Most of the top ranked factor means in the $L = 15$ basis function case were also identified as top ranked functions in the $L = 10$ and $L = 20$ case (see Supplementary Material).

\begin{figure}[t!]
  \begin{center}
		\includegraphics[width=\textwidth]{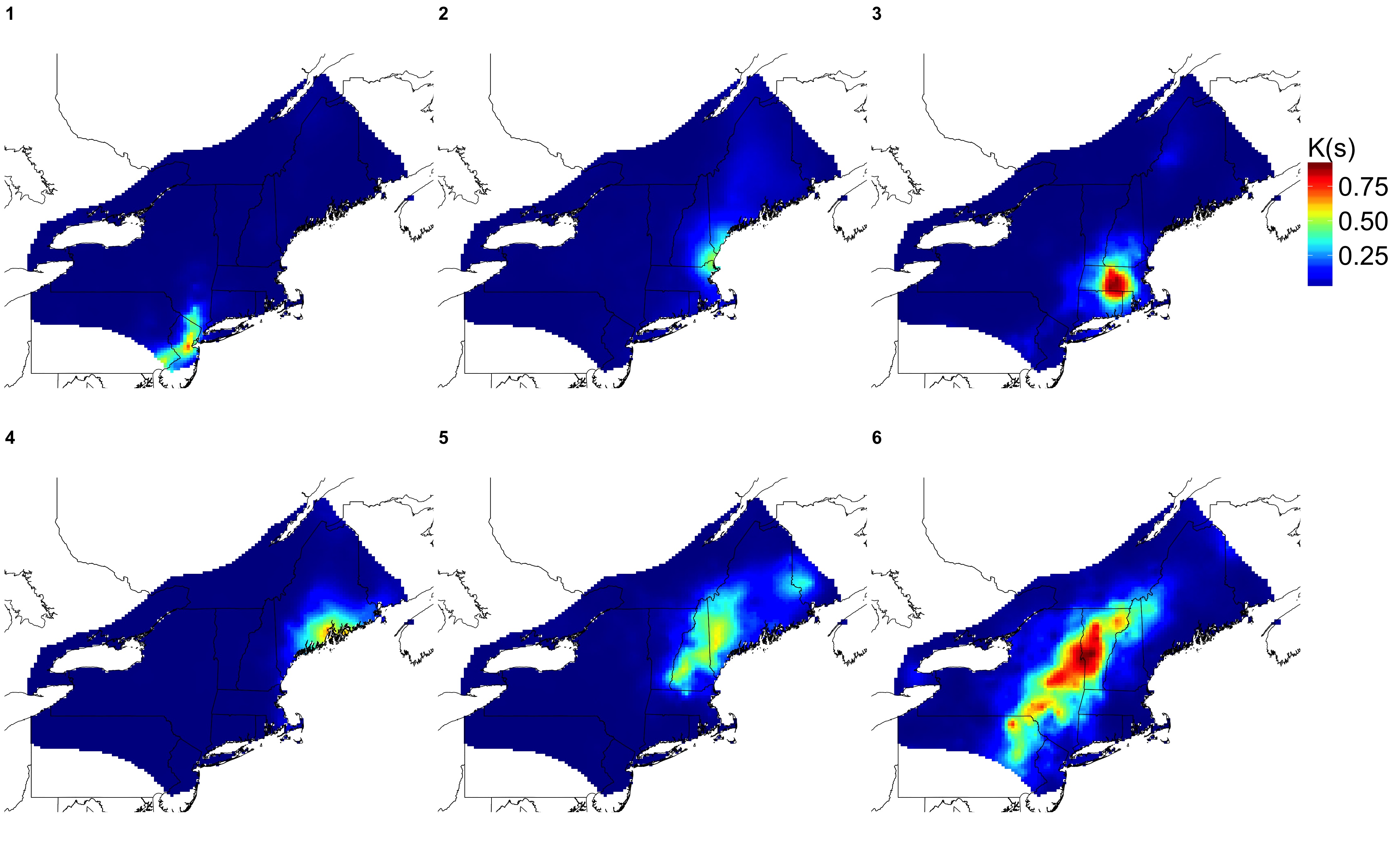}
		\caption{First six spatial basis functions ordered by the variance of their corresponding random basis coefficients from largest to smallest (left to right, top to bottom) for the $L = 15$ basis function model. The year-to-year variation among the coefficients of these first six basis functions accounts for 97\% of the total year-to-year variation among all of the basis coefficients. The shapes of the latent factors have reasonable interpretations in terms of geographic coastal and mountain features.}
		\label{fig:top_K}
	\end{center}
\end{figure}
Unlike the pointwise marginal surfaces, which do not provide any information about the joint dependence of extremes, these basis functions capture spatial regions of simultaneous (in this case, merely the same year) extreme precipitation. The proportion of the total variation among the $A_{l,t}$ accounted for by variation in the coefficients of each of the first six basis functions is $0.48, 0.33, 0.07, 0.04, 0.03,$ and $0.02$ respectively. This does not imply that the top ranked factor is the dominating kernel $48\%$ of the time. Rather, if the variance of the scaling coefficients for the $l$th factor is high, then the year-to-year differences in the spatial modes of extremes should be well described by the peaks and troughs of the $l$th factor. For example, if $K_l(\bs)$ has a peak around some location $\bs^*$ then the conditional GEV distribution (given the factors and scaling coefficients) will be stochastically larger at $\bs^*$ in years when $A_{l,t}$ is large and smaller when $A_{l,t}$ is small. Therefore, the low ranked factors describe regions where precipitation tends to be extreme together or more moderate together. The latent factors in Figure \ref{fig:top_K} have reasonable physical interpretations that are reflective of natural geographic features. In particular, they resemble observed patterns in extreme precipitation events occurring along the coast and mountain range borders. Just as with spatial PCA, we hesitate to make strong interpretations of the identified factors. However, the first four factors appear to correspond to coastal regions in New Jersey and New England that tend to be affected by the same localized tropical cyclones and convective storms. The last two modes appear to reflect the orographic effect of mountains in cloud formation, which cause moist air to rise.

\section{DISCUSSION} 
\label{sec:discussion}
In this paper, we extend the max-stable model for spatial extremes developed by \citet{Reich12} in several ways. First, by using flexible log-Gaussian process basis functions, our model provides a more realistic low-dimensional factor representation that can be used to visualize the main modes of spatial variability among extremes. Second, our approach relaxes the rigid spatial dependence structure imposed by max-stable models, while possessing the positive dependence inherent to distributions for maxima. Inference on the tail dependence class is also possible, as our model can capture asymptotic independence when $\theta >0$, while having an asymptotically dependent, max-stable model on the boundary of the parameter space (when $\theta = 0$). 

We apply our model to extreme precipitation over the northeastern United States and Canada. Because it accounts for the spatial dependence among maxima and we are able to efficiently make conditional draws from our fitted model. The precipitation predictions from our model could be incorporated into a hydrological model for the flow path dynamics that incorporates factors like drainage basin topography, land use, and land cover to describe how precipitation falling over a common catchment translates into drainage and potential flooding. The precipitation analysis does not account for the cumulative effect of heavy precipitation over several days, which can overload an urban stormwater drainage system that is already operating at capacity. Further temporal modeling of the marginal distributions and space-time dependence characteristics would facilitate such an analysis; see, e.g., \citet{Huser14} for space-time modeling of precipitation extremes using max-stable processes.

For future work, adding a point mass at $\theta = 0$ in the prior and proposal distributions would make it possible to account for model uncertainty and simultaneously perform model selection directly within the MCMC. Finally, while our focus in this paper has been on flexible sub-asymptotic modeling of maxima, another avenue for research is to investigate relaxing the rigid dependence structure of limiting generalized Pareto process models for peaks-over-threshold data \citep[see, e.g.,][]{Castro18,HuserWadsworth19}.

\appendix

\section{Model Tail Dependence Properties}
\label{sec:tail_dep_properties}
Since the marginal distributions of $Z(\bs)$ are the same when constructed using the log-Gaussian process basis, $Z(\bs_1)$ and $Z(\bs_2)$ are asymptotically independent if $\pr(Z(\bs_1) >z|Z(\bs_2)>z) \rightarrow 0$ as $z \rightarrow \infty$. The marginal distribution of the process at location $\bs$ conditional on the basis functions is $G_{\bs}\{z|K_l(\bs), l = 1,\ldots L\} = \exp(L \theta^\alpha - \sum_{l= 1}^L[\theta + \{{K_l(\bs)/z}\}^{1/\alpha}]^\alpha)$, and the joint distribution at two locations $\bs_1$ and $\bs_2$ is
\[
\pr\{Z(\bs_1) \leq z_1, Z(\bs_2)\leq z_2|K_l(\bs), l = 1,\ldots L\} = \exp\left(L \theta^\alpha - \sum_{l= 1}^L\left[\theta + \left\{{K_l(\bs_1)\over z_1}\right\}^{1/\alpha} + \left\{{K_l(\bs_2)\over z_2}\right\}^{1/\alpha}\right]^\alpha\right).
\]
For brevity, we will drop the indices $l = 1,\ldots L$, and write, e.g., $G_{\bs}\{z|K_l(\bs)\} \equiv G_{\bs}\{z|K_l(\bs), l = 1,\ldots L\}$. By $L'Hospital$'s rule, we obtain
{\footnotesize \begin{align}
&\chi(\bs_1, \bs_2)|K_l(\bs) = 1 +\underset{z \rightarrow \infty}{\lim} {{{\rm d}\over {\rm d}z} G_{\bs_1}\{z|K_l(\bs)\}\over{{\rm d}\over {\rm d}z}G_{\bs_2}\{z|K_l(\bs)\}} - \underset{z \rightarrow \infty}{\lim} {{{\rm d}\over {\rm d}z}\Pr\{Z(\bs_1)\leq z, Z(\bs_2)\leq z|K_l(\bs)\}\over{{\rm d} \over {\rm d}z}G_{\bs_2}\{z|K_l(\bs)\}} \label{eq1} \nonumber \\
&= 2 - \underset{z \rightarrow \infty}{\lim} {\Pr\{Z(\bs_1)\leq z, Z(\bs_2)\leq z|K_l(\bs)\}\over G_{\bs_2}\{z|K_l(\bs)\}}
\underset{z \rightarrow \infty}{\lim} {
\sum_{l= 1}^L\left[\theta + \left\{{K_l(\bs_1)\over z}\right\}^{1/\alpha} + \left\{{K_l(\bs_2)\over z}\right\}^{1/\alpha}\right]^{\alpha-1}\left\{K_l(\bs_1)^{1/\alpha} + K_l(\bs_2)^{1/\alpha}\right\}
\over \sum_{l= 1}^L\left[\theta + \left\{{K_l(\bs_2)\over z}\right\}^{1/\alpha}\right]^{\alpha-1}K_l(\bs_2)^{1/\alpha}} \nonumber \\
&= 0 \nonumber
\end{align}}
when $\theta >0$. Finally, by application of the Dominated Convergence Theorem, since $|\chi_z(\bs_1, \bs_2)|K_l(\bs)| < 1$, we obtain $\chi(\bs_1, \bs_2) = {\rm E}\{\underset{z \rightarrow \infty}{\lim}  \chi_z(\bs_1, \bs_2)|K_l(\bs)\}= 0$ for all $\bs_1, \bs_2 \in \mathcal{S}$ For more detail, see the Supplementary Material.

In the case of $\theta = 0$ and $\alpha < 1$, \cite{Reich12} showed that $Z(\bs)\mid K_l(\bs),$ is max-stable with extremal coefficient $\theta_2(\bs_1, \bs_2)\mid K_l(\bs) = \sum_{l = 1}^L\left[K_l(\bs_1)^{1/\alpha} + K_l(\bs_2)^{1/\alpha}\right]^\alpha$. Using the relation for max-stable processes with unit Fr\'{e}chet margins that $\chi(\bs_1, \bs_2) = 2 - \theta(\bs_1, \bs_2)$, and by the Dominated Convergence Theorem, we have $\chi(\bs_1, \bs_2) = {\rm E}\{\underset{z \rightarrow \infty}{\lim}  \chi_z(\bs_1, \bs_2)|K_l(\bs)\}= 2 - {\rm E}\{\theta_2(\bs_1, \bs_2)\mid K_l(\bs)\} = 2 - {\rm E}\{\sum_{l = 1}^L\left[K_l(\bs_1)^{1/\alpha} + K_l(\bs_2)^{1/\alpha}\right]^\alpha\} > 0$ when $\alpha < 1$ for all $\bs_1, \bs_2 \in \mathcal{S}$. So, when $\theta = 0$ and $\alpha < 1$, $Z(\bs)$ is asymptotically dependent, both conditionally on $K_l, l = 1, \ldots, L$ and unconditionally.

\section*{Acknowledgement}
The authors gratefully acknowledge the support of NSF grant DMS-1752280 as well as seed grants from the Institute for CyberScience and the Institute for Energy and the Environment at Pennsylvania State University.  Computations for this research were performed on the Pennsylvania State University’s Institute for CyberScience Advanced CyberInfrastructure (ICS-ACI). This content is solely the responsibility of the authors and does not necessarily represent the views of the Institute for CyberScience.

\bibliographystyle{chicago}
\bibliography{bib}

\begin{thebibliography}{}

\bibitem[\protect\citeauthoryear{Balkema and Resnick}{Balkema and
  Resnick}{1977}]{Balkema77}
Balkema, A.~A. and S.~I. Resnick (1977).
\newblock Max-infinite divisibility.
\newblock {\em Journal of Applied Probability\/}~{\em 14\/}(2), 309--319.

\bibitem[\protect\citeauthoryear{Brown and Resnick}{Brown and
  Resnick}{1977}]{Brown77}
Brown, B.~M. and S.~I. Resnick (1977).
\newblock Extreme values of independent stochastic processes.
\newblock {\em Journal of Applied Probability\/}~{\em 14\/}(4), 732--739.

\bibitem[\protect\citeauthoryear{Castro~Camilo and Huser}{Castro~Camilo and
  Huser}{2018}]{Castro18}
Castro~Camilo, D. and R.~Huser (2018).
\newblock {Local likelihood estimation of complex tail dependence structures,
  applied to U.S. precipitation extremes}.
\newblock {arXiv preprint 1710.00875}.

\bibitem[\protect\citeauthoryear{Castruccio, Huser, and Genton}{Castruccio
  et~al.}{2016}]{Castruccio16}
Castruccio, S., R.~Huser, and M.~G. Genton (2016).
\newblock High-order composite likelihood inference for max-stable
  distributions and processes.
\newblock {\em Journal of Computational and Graphical Statistics\/}~{\em
  25\/}(4), 1212--1229.

\bibitem[\protect\citeauthoryear{Coles, Heffernan, and Tawn}{Coles
  et~al.}{1999}]{Coles99}
Coles, S., J.~Heffernan, and J.~Tawn (1999).
\newblock Dependence measures for extreme value analyses.
\newblock {\em Extremes\/}~{\em 2\/}(4), 339.

\bibitem[\protect\citeauthoryear{Cooley, Nychka, and Naveau}{Cooley
  et~al.}{2007}]{Cooley07}
Cooley, D., D.~Nychka, and P.~Naveau (2007).
\newblock Bayesian spatial modeling of extreme precipitation return levels.
\newblock {\em Journal of the American Statistical Association\/}~{\em
  102\/}(479), 824--840.

\bibitem[\protect\citeauthoryear{Cooley and Thibaud}{Cooley and
  Thibaud}{2018}]{Cooley18}
Cooley, D. and E.~Thibaud (2018).
\newblock Decompositions of dependence for high-dimensional extremes.
\newblock {\em arXiv preprint 1612.07190\/}.

\bibitem[\protect\citeauthoryear{Crowder}{Crowder}{1989}]{Crowder89}
Crowder, M. (1989).
\newblock A multivariate distribution with weibull connections.
\newblock {\em Journal of the Royal Statistical Society: Series B
  (Methodological)\/}~{\em 51\/}(1), 93--107.

\bibitem[\protect\citeauthoryear{Davison and Huser}{Davison and
  Huser}{2015}]{Davison15}
Davison, A.~C. and R.~Huser (2015).
\newblock Statistics of extremes.
\newblock {\em Annual Review of Statistics and Its Application\/}~{\em 2},
  203--235.

\bibitem[\protect\citeauthoryear{Davison, Huser, and Thibaud}{Davison
  et~al.}{2013}]{Davison13}
Davison, A.~C., R.~Huser, and E.~Thibaud (2013).
\newblock Geostatistics of dependent and asymptotically independent extremes.
\newblock {\em Mathematical Geosciences\/}~{\em 45\/}(5), 511--529.

\bibitem[\protect\citeauthoryear{Davison, Huser, and Thibaud}{Davison
  et~al.}{2019}]{Davison19}
Davison, A.~C., R.~Huser, and E.~Thibaud (2019).
\newblock Spatial extremes.
\newblock In A.~E. Gelfand, M.~Fuentes, and R.~L. Smith (Eds.), {\em Handbook
  of Environmental and Ecological Statistics}, pp.\  711--744. CRC Press.

\bibitem[\protect\citeauthoryear{Davison, Padoan, and Ribatet}{Davison
  et~al.}{2012}]{Davison12}
Davison, A.~C., S.~A. Padoan, and M.~Ribatet (2012).
\newblock Statistical modeling of spatial extremes.
\newblock {\em Statistical Science. A Review Journal of the Institute of
  Mathematical Statistics\/}~{\em 27\/}(2), 161--186.

\bibitem[\protect\citeauthoryear{de~Haan and Ferreira}{de~Haan and
  Ferreira}{2006}]{deHaan07}
de~Haan, L. and A.~Ferreira (2006).
\newblock {\em Extreme value theory}.
\newblock Springer Series in Operations Research and Financial Engineering.
  Springer, New York.
\newblock An introduction.

\bibitem[\protect\citeauthoryear{Demsar, Harris, Brunsdon, Fotheringham, and
  McLoone}{Demsar et~al.}{2013}]{Demsar13}
Demsar, U., P.~Harris, C.~Brunsdon, A.~S. Fotheringham, and S.~McLoone (2013).
\newblock Principal component analysis on spatial data: an overview.
\newblock {\em Annals of the Association of American Geographers\/}~{\em
  103\/}(1), 106--128.

\bibitem[\protect\citeauthoryear{Devroye}{Devroye}{2009}]{Devroye09}
Devroye, L. (2009).
\newblock Random variate generation for exponentially and polynomially tilted
  stable distributions.
\newblock {\em ACM Transactions on Modeling and Computer Simulation
  (TOMACS)\/}~{\em 19\/}(4), 18.

\bibitem[\protect\citeauthoryear{Dombry, Engelke, and Oesting}{Dombry
  et~al.}{2017}]{Dombry17}
Dombry, C., S.~Engelke, and M.~Oesting (2017).
\newblock Bayesian inference for multivariate extreme value distributions.
\newblock {\em Electronic Journal of Statistics\/}~{\em 11\/}(2), 4813--4844.

\bibitem[\protect\citeauthoryear{Dombry, \'Eyi-Minko, and Ribatet}{Dombry
  et~al.}{2013}]{Dombry13}
Dombry, C., F.~\'Eyi-Minko, and M.~Ribatet (2013).
\newblock Conditional simulation of max-stable processes.
\newblock {\em Biometrika\/}~{\em 100\/}(1), 111--124.

\bibitem[\protect\citeauthoryear{Ferreira and de~Haan}{Ferreira and
  de~Haan}{2014}]{Ferreira14}
Ferreira, A. and L.~de~Haan (2014).
\newblock The generalized {P}areto process; with a view towards application and
  simulation.
\newblock {\em Bernoulli. Official Journal of the Bernoulli Society for
  Mathematical Statistics and Probability\/}~{\em 20\/}(4), 1717--1737.

\bibitem[\protect\citeauthoryear{Foug\`eres, Nolan, and Rootz\'en}{Foug\`eres
  et~al.}{2009}]{Fougeres09}
Foug\`eres, A.-L., J.~P. Nolan, and H.~Rootz\'en (2009).
\newblock Models for dependent extremes using stable mixtures.
\newblock {\em Scandinavian Journal of Statistics. Theory and
  Applications\/}~{\em 36\/}(1), 42--59.

\bibitem[\protect\citeauthoryear{Gneiting and Raftery}{Gneiting and
  Raftery}{2007}]{Gneiting07}
Gneiting, T. and A.~E. Raftery (2007).
\newblock Strictly proper scoring rules, prediction, and estimation.
\newblock {\em Journal of the American Statistical Association\/}~{\em
  102\/}(477), 359--378.

\bibitem[\protect\citeauthoryear{Hannachi, Jolliffe, and Stephenson}{Hannachi
  et~al.}{2007}]{Hannachi07}
Hannachi, A., I.~Jolliffe, and D.~Stephenson (2007).
\newblock Empirical orthogonal functions and related techniques in atmospheric
  science: A review.
\newblock {\em International journal of climatology\/}~{\em 27\/}(9),
  1119--1152.

\bibitem[\protect\citeauthoryear{Hougaard}{Hougaard}{1986}]{Hougaard86}
Hougaard, P. (1986).
\newblock Survival models for heterogeneous populations derived from stable
  distributions.
\newblock {\em Biometrika\/}~{\em 73\/}(2), 387--396.

\bibitem[\protect\citeauthoryear{Huser and Davison}{Huser and
  Davison}{2014}]{Huser14}
Huser, R. and A.~C. Davison (2014).
\newblock Space-time modelling of extreme events.
\newblock {\em Journal of the Royal Statistical Society. Series B. Statistical
  Methodology\/}~{\em 76\/}(2), 439--461.

\bibitem[\protect\citeauthoryear{Huser, Dombry, Ribatet, and Genton}{Huser
  et~al.}{2019}]{Huser19}
Huser, R., C.~Dombry, M.~Ribatet, and M.~G. Genton (2019).
\newblock Full likelihood inference for max-stable data.
\newblock {\em Stat\/}~{\em 8}, e218.

\bibitem[\protect\citeauthoryear{Huser, Opitz, and Thibaud}{Huser
  et~al.}{2017}]{Huser17}
Huser, R., T.~Opitz, and E.~Thibaud (2017).
\newblock Bridging asymptotic independence and dependence in spatial extremes
  using {G}aussian scale mixtures.
\newblock {\em Spatial Statistics\/}~{\em 21\/}(part A), 166--186.

\bibitem[\protect\citeauthoryear{Huser, Opitz, and Thibaud}{Huser
  et~al.}{2018}]{Huser18}
Huser, R., T.~Opitz, and E.~Thibaud (2018).
\newblock Max-infinitely divisible models and inference for spatial extremes.
\newblock arXiv preprint 1801.02946.

\bibitem[\protect\citeauthoryear{Huser and Wadsworth}{Huser and
  Wadsworth}{2019}]{HuserWadsworth19}
Huser, R. and J.~L. Wadsworth (2019).
\newblock Modeling spatial processes with unknown extremal dependence class.
\newblock {\em Journal of the American Statistical Association\/}~{\em 114},
  434--444.

\bibitem[\protect\citeauthoryear{Jolliffe}{Jolliffe}{2002}]{Jolliffe02}
Jolliffe, I.~T. (2002).
\newblock {\em Principal component analysis\/} (Second ed.).
\newblock Springer Series in Statistics. Springer-Verlag, New York.

\bibitem[\protect\citeauthoryear{Kabluchko and Schlather}{Kabluchko and
  Schlather}{2010}]{Kabluchko10}
Kabluchko, Z. and M.~Schlather (2010).
\newblock Ergodic properties of max-infinitely divisible processes.
\newblock {\em Stochastic Processes and their Applications\/}~{\em 120\/}(3),
  281--295.

\bibitem[\protect\citeauthoryear{Kabluchko, Schlather, and de~Haan}{Kabluchko
  et~al.}{2009}]{Kabluchko09}
Kabluchko, Z., M.~Schlather, and L.~de~Haan (2009).
\newblock Stationary max-stable fields associated to negative definite
  functions.
\newblock {\em The Annals of Probability\/}~{\em 37\/}(5), 2042--2065.

\bibitem[\protect\citeauthoryear{Kabluchko and Stoev}{Kabluchko and
  Stoev}{2016}]{Kabluchko16}
Kabluchko, Z. and S.~Stoev (2016).
\newblock Stochastic integral representations and classification of sum- and
  max-infinitely divisible processes.
\newblock {\em Bernoulli. Official Journal of the Bernoulli Society for
  Mathematical Statistics and Probability\/}~{\em 22\/}(1), 107--142.

\bibitem[\protect\citeauthoryear{Kwiatkowski, Phillips, Schmidt, and
  Shin}{Kwiatkowski et~al.}{1992}]{Kwiatkowski92}
Kwiatkowski, D., P.~C. Phillips, P.~Schmidt, and Y.~Shin (1992).
\newblock Testing the null hypothesis of stationarity against the alternative
  of a unit root: How sure are we that economic time series have a unit root?
\newblock {\em Journal of econometrics\/}~{\em 54\/}(1-3), 159--178.

\bibitem[\protect\citeauthoryear{Marin, Mengersen, and Robert}{Marin
  et~al.}{2005}]{Marin05}
Marin, J.-M., K.~Mengersen, and C.~P. Robert (2005).
\newblock Bayesian modelling and inference on mixtures of distributions.
\newblock In {\em Bayesian thinking: modeling and computation}, Volume~25 of
  {\em Handbook of Statistics}, pp.\  459--507. Elsevier/North-Holland,
  Amsterdam.

\bibitem[\protect\citeauthoryear{Marin and Robert}{Marin and
  Robert}{2007}]{Marin07}
Marin, J.-M. and C.~P. Robert (2007).
\newblock {\em Bayesian core: a practical approach to computational {B}ayesian
  statistics}.
\newblock Springer Texts in Statistics. Springer, New York.

\bibitem[\protect\citeauthoryear{Morris}{Morris}{2016}]{Morris16}
Morris, S.~A. (2016).
\newblock {\em Spatial {M}ethods for {M}odeling {E}xtreme and {R}are {E}vents}.
\newblock Thesis (Ph.D.) North Carolina State University.

\bibitem[\protect\citeauthoryear{Opitz, Huser, Bakka, and Rue}{Opitz
  et~al.}{2018}]{Opitz18}
Opitz, T., R.~Huser, H.~Bakka, and H.~Rue (2018).
\newblock {INLA goes extreme: Bayesian tail regression for the estimation of
  high spatio-temporal quantiles}.
\newblock {\em Extremes\/}~{\em 21\/}(3), 441--462.

\bibitem[\protect\citeauthoryear{Padoan}{Padoan}{2013}]{Padoan13}
Padoan, S.~A. (2013).
\newblock Extreme dependence models based on event magnitude.
\newblock {\em Journal of Multivariate Analysis\/}~{\em 122}, 1--19.

\bibitem[\protect\citeauthoryear{Perica, Martin, Pavlovic, Roy, Laurent,
  Trypaluk, Unruh, Yekta, and Bonnin}{Perica et~al.}{2013}]{Perica13}
Perica, S., D.~Martin, S.~Pavlovic, I.~Roy, M.~S. Laurent, C.~Trypaluk,
  D.~Unruh, M.~Yekta, and G.~Bonnin (2013).
\newblock Noaa atlas 14 volume 9 version 2, precipitation-frequency atlas of
  the united states, southeastern states.
\newblock {\em NOAA, National Weather Service\/}~{\em 9}, 18.

\bibitem[\protect\citeauthoryear{Reich and Shaby}{Reich and
  Shaby}{2012}]{Reich12}
Reich, B.~J. and B.~A. Shaby (2012).
\newblock A hierarchical max-stable spatial model for extreme precipitation.
\newblock {\em The Annals of Applied Statistics\/}~{\em 6\/}(4), 1430--1451.

\bibitem[\protect\citeauthoryear{Resnick}{Resnick}{1987}]{Resnick87}
Resnick, S.~I. (1987).
\newblock {\em Extreme values, regular variation, and point processes},
  Volume~4 of {\em Applied Probability. A Series of the Applied Probability
  Trust}.
\newblock Springer-Verlag, New York.

\bibitem[\protect\citeauthoryear{Robins, van~der Vaart, and Ventura}{Robins
  et~al.}{2000}]{Robins00}
Robins, J.~M., A.~van~der Vaart, and V.~Ventura (2000).
\newblock Asymptotic distribution of {$p$} values in composite null models.
\newblock {\em J. Amer. Statist. Assoc.\/}~{\em 95\/}(452), 1143--1167,
  1171--1172.
\newblock With comments and a rejoinder by the authors.

\bibitem[\protect\citeauthoryear{Sang and Gelfand}{Sang and
  Gelfand}{2010}]{Sang10}
Sang, H. and A.~E. Gelfand (2010).
\newblock Continuous spatial process models for spatial extreme values.
\newblock {\em Journal of Agricultural, Biological, and Environmental
  Statistics\/}~{\em 15\/}(1), 49--65.

\bibitem[\protect\citeauthoryear{Schlather and Tawn}{Schlather and
  Tawn}{2003}]{Schlather03}
Schlather, M. and J.~A. Tawn (2003).
\newblock A dependence measure for multivariate and spatial extreme values:
  properties and inference.
\newblock {\em Biometrika\/}~{\em 90\/}(1), 139--156.

\bibitem[\protect\citeauthoryear{Stephenson}{Stephenson}{2009}]{Stephenson09}
Stephenson, A.~G. (2009).
\newblock High-dimensional parametric modelling of multivariate extreme events.
\newblock {\em Australian \& New Zealand Journal of Statistics\/}~{\em
  51\/}(1), 77--88.

\bibitem[\protect\citeauthoryear{Suro, Firda, and Szabo}{Suro
  et~al.}{2009}]{Suro09}
Suro, T.~P., G.~D. Firda, and C.~O. Szabo (2009).
\newblock Flood of june 26-29, 2006, mohawk, delaware, and susquehanna river
  basins, new york.
\newblock {\em US Geological Survey Open-File Report\/}~{\em 1063}, 354.

\bibitem[\protect\citeauthoryear{Tawn}{Tawn}{1990}]{Tawn90}
Tawn, J.~A. (1990).
\newblock Modelling multivariate extreme value distributions.
\newblock {\em Biometrika\/}~{\em 77\/}(2), 245--253.

\bibitem[\protect\citeauthoryear{Thibaud and Opitz}{Thibaud and
  Opitz}{2015}]{Thibaud15}
Thibaud, E. and T.~Opitz (2015).
\newblock Efficient inference and simulation for elliptical {P}areto processes.
\newblock {\em Biometrika\/}~{\em 102\/}(4), 855--870.

\bibitem[\protect\citeauthoryear{Wadsworth and Tawn}{Wadsworth and
  Tawn}{2012}]{Wadsworth12}
Wadsworth, J.~L. and J.~A. Tawn (2012).
\newblock Dependence modelling for spatial extremes.
\newblock {\em Biometrika\/}~{\em 99\/}(2), 253--272.

\end{thebibliography}


\begin{thebibliography}{}

\bibitem[\protect\citeauthoryear{Balkema and Resnick}{Balkema and
  Resnick}{1977}]{Balkema77}
Balkema, A.~A. and S.~I. Resnick (1977).
\newblock Max-infinite divisibility.
\newblock {\em Journal of Applied Probability\/}~{\em 14\/}(2), 309--319.

\bibitem[\protect\citeauthoryear{Devroye}{Devroye}{2009}]{Devroye09}
Devroye, L. (2009).
\newblock Random variate generation for exponentially and polynomially tilted
  stable distributions.
\newblock {\em ACM Transactions on Modeling and Computer Simulation
  (TOMACS)\/}~{\em 19\/}(4), 18.

\bibitem[\protect\citeauthoryear{Huser, Opitz, and Thibaud}{Huser
  et~al.}{2018}]{Huser18}
Huser, R., T.~Opitz, and E.~Thibaud (2018).
\newblock Max-infinitely divisible models and inference for spatial extremes.
\newblock arXiv preprint 1801.02946.

\bibitem[\protect\citeauthoryear{Ibragimov and Chernin}{Ibragimov and
  Chernin}{1959}]{Chernin59}
Ibragimov, I.~A. and K.~E. Chernin (1959).
\newblock On the unimodality of stable laws.
\newblock {\em Theory of Probability and its Applications\/}~{\em 4}, 417--419.

\bibitem[\protect\citeauthoryear{Kanter}{Kanter}{1975}]{Kanter75}
Kanter, M. (1975).
\newblock Stable densities under change of scale and total variation
  inequalities.
\newblock {\em The Annals of Probability\/}~{\em 3\/}(4), 697--707.

\bibitem[\protect\citeauthoryear{Kuss and Rasmussen}{Kuss and
  Rasmussen}{2005}]{Kuss05}
Kuss, M. and C.~E. Rasmussen (2005).
\newblock Assessing approximate inference for binary {G}aussian process
  classification.
\newblock {\em Journal of Machine Learning Research (JMLR)\/}~{\em 6},
  1679--1704.

\bibitem[\protect\citeauthoryear{Resnick}{Resnick}{1987}]{Resnick87}
Resnick, S.~I. (1987).
\newblock {\em Extreme values, regular variation, and point processes},
  Volume~4 of {\em Applied Probability. A Series of the Applied Probability
  Trust}.
\newblock Springer-Verlag, New York.

\bibitem[\protect\citeauthoryear{Shaby and Wells}{Shaby and
  Wells}{2010}]{Shaby10}
Shaby, B. and M.~Wells (2010).
\newblock Exploring an adaptive {M}etropolis algorithm.
\newblock Technical Report 1011-14, Duke University Department of Stastical
  Science.

\bibitem[\protect\citeauthoryear{Stephenson}{Stephenson}{2009}]{Stephenson09}
Stephenson, A.~G. (2009).
\newblock High-dimensional parametric modelling of multivariate extreme events.
\newblock {\em Australian \& New Zealand Journal of Statistics\/}~{\em
  51\/}(1), 77--88.

\end{thebibliography}

\end{document}


\maketitle

\section{ASYMPTOTIC INDEPENDENCE OF THE MAX-ID MODEL}
\label{sec:properties_of_the_max_id_model}
Since the marginal distributions of $Z(\bs)$ are the same when constructed using the log-Gaussian process basis, $Z(\bs_1)$ and $Z(\bs_2)$ are asymptotically independent if $\pr(Z(\bs_1) >z|Z(\bs_2)>z) \rightarrow 0$ as $z \rightarrow \infty$. The marginal distribution of the process at location $\bs$ conditional on the basis functions is
\[
G_{\bs}\{z|K_l(\bs), l = 1,\ldots L\} = \exp\left(L \theta^\alpha - \sum_{l= 1}^L\left[\theta + \left\{{K_l(\bs)\over z}\right\}^{1/\alpha}\right]^\alpha\right),
\]
and the joint distribution at two locations $\bs_1$ and $\bs_2$ is
\[
\pr\{Z(\bs_1) \leq z_1, Z(\bs_2)\leq z_2|K_l(\bs), l = 1,\ldots L\} = \exp\left(L \theta^\alpha - \sum_{l= 1}^L\left[\theta + \left\{{K_l(\bs_1)\over z_1}\right\}^{1/\alpha} + \left\{{K_l(\bs_2)\over z_2}\right\}^{1/\alpha}\right]^\alpha\right).
\]
For brevity, we will drop the indices $l = 1,\ldots L$, and write, e.g., $G_{\bs}\{z|K_l(\bs)\} \equiv G_{\bs}\{z|K_l(\bs), l = 1,\ldots L\}$. We have that $\underset{z \rightarrow \infty}{\lim}1 - G_{\bs_1}\{z|K_l(\bs)\}-G_{\bs_2}\{z|K_l(\bs)\} + \Pr\{Z(\bs_1)\leq z, Z(\bs_2)\leq z|K_l(\bs)\} =0$, and $\underset{z \rightarrow \infty}{\lim} 1-G_{\bs_2}\{z|K_l(\bs)\} = 0$. Thus, by $L'Hospital$'s rule, we obtain
\begin{align}
\chi(\bs_1, \bs_2)|K_l(\bs)  &= \underset{z \rightarrow \infty}{\lim} \Pr\{Z(\bs_1) >z| Z(\bs_2)>z, K_l(\bs)\} \nonumber \\
&= \underset{z \rightarrow \infty}{\lim} {1 - G_{\bs_1}\{z|K_l(\bs)\} - G_{\bs_2}\{z|K_l(\bs)\} + \Pr\{Z(\bs_1)\leq z, Z(\bs_2)\leq z|K_l(\bs)\}\over 1 - G_{\bs_2}\{z|K_l(\bs)\}} \nonumber\\ 
&= \underset{z \rightarrow \infty}{\lim}{{{\rm d}\over {\rm d}z}[1 - G_{\bs_2}\{z|K_l(\bs)\}-G_{\bs_1}\{z|K_l(\bs)\} + \Pr\{Z(\bs_1)\leq z, Z(\bs_2)\leq z|K_l(\bs)\}]\over {{\rm d}\over {\rm d}z}[1 - G_{\bs_2}\{z|K_l(\bs)\}]} \nonumber\\
&= 1 +\underset{z \rightarrow \infty}{\lim} {{{\rm d}\over {\rm d}z} G_{\bs_1}\{z|K_l(\bs)\}\over{{\rm d}\over {\rm d}z}G_{\bs_2}\{z|K_l(\bs)\}} - \underset{z \rightarrow \infty}{\lim} {{{\rm d}\over {\rm d}z}\Pr\{Z(\bs_1)\leq z, Z(\bs_2)\leq z|K_l(\bs)\}\over{{\rm d} \over {\rm d}z}G_{\bs_2}\{z|K_l(\bs)\}}. \label{eq1}
\end{align}
Now for $\theta >0$, the second term in (\ref{eq1}) becomes
\begin{align}
\underset{z \rightarrow \infty}{\lim} {{{\rm d}\over {\rm d}z}G_{\bs_1}\{z|K_l(\bs)\}\over{{\rm d}\over {\rm d}z}G_{\bs_2}\{z|K_l(\bs)\}}  
&= \underset{z \rightarrow \infty}{\lim}  {G_{\bs_1}\{z|K_l(\bs)\}\left(\alpha z^{-(1+1/\alpha)}\sum_{l= 1}^L\left[\theta + \left\{{K_l(\bs_1)\over z}\right\}^{1/\alpha}\right]^{\alpha-1}K_l(\bs_1)^{1/\alpha}\right)\over G_{\bs_2}\{z|K_l(\bs)\}\left(\alpha z^{-(1+1/\alpha)}\sum_{l= 1}^L\left[\theta + \left\{{K_l(\bs_2)\over z}\right\}^{1/\alpha}\right]^{\alpha-1}K_l(\bs_2)^{1/\alpha}\right)}\nonumber \\
&= \underset{z \rightarrow \infty}{\lim} {G_{\bs_1}\{z|K_l(\bs)\}\over G_{\bs_2}\{z|K_l(\bs)\}}\times \underset{z \rightarrow \infty}{\lim} {\sum_{l= 1}^L\left[\theta + \left\{{K_l(\bs_1)\over z}\right\}^{1/\alpha}\right]^{\alpha-1}K_l(\bs_1)^{1/\alpha}\over \sum_{l= 1}^L\left[\theta + \left\{{K_l(\bs_2)\over z}\right\}^{1/\alpha}\right]^{\alpha-1}K_l(\bs_2)^{1/\alpha}}\nonumber\\
&= {\sum_{l= 1}^LK_l(\bs_1)^{1/\alpha}\over \sum_{l= 1}^L K_l(\bs_2)^{1/\alpha}},\label{eq:firstpart}
\end{align}
and the third term in (\ref{eq1}) becomes
{\footnotesize\begin{align}
&\underset{z \rightarrow \infty}{\lim} {{{\rm d}\over {\rm d}z}\Pr\{Z(\bs_1)\leq z, Z(\bs_2)\leq z|K_l(\bs)\}\over {{\rm d}\over {\rm d}z}G_{\bs_2}\{z|K_l(\bs)\}}\nonumber\\
&= \underset{z \rightarrow \infty}{\lim} {\Pr\{Z(\bs_1)\leq z, Z(\bs_2)\leq z|K_l(\bs)\}
\left(\alpha z^{-(1+1/\alpha)}\sum_{l= 1}^L\left[\theta + \left\{{K_l(\bs_1)\over z}\right\}^{1/\alpha} + \left\{{K_l(\bs_2)\over z}\right\}^{1/\alpha}\right]^{\alpha-1}\left\{K_l(\bs_1)^{1/\alpha} + K_l(\bs_2)^{1/\alpha}\right\}\right)
\over G_{\bs_2}\{z|K_l(\bs)\}\left(\alpha z^{-(1+1/\alpha)}\sum_{l= 1}^L\left[\theta + \left\{{K_l(\bs_2)\over z}\right\}^{1/\alpha}\right]^{\alpha-1}K_l(\bs_2)^{1/\alpha}\right)} \nonumber\\
&= \underset{z \rightarrow \infty}{\lim} {\Pr\{Z(\bs_1)\leq z, Z(\bs_2)\leq z|K_l(\bs)\}\over G_{\bs_2}\{z|K_l(\bs)\}}
\underset{z \rightarrow \infty}{\lim} {
\sum_{l= 1}^L\left[\theta + \left\{{K_l(\bs_1)\over z}\right\}^{1/\alpha} + \left\{{K_l(\bs_2)\over z}\right\}^{1/\alpha}\right]^{\alpha-1}\left\{K_l(\bs_1)^{1/\alpha} + K_l(\bs_2)^{1/\alpha}\right\}
\over \sum_{l= 1}^L\left[\theta + \left\{{K_l(\bs_2)\over z}\right\}^{1/\alpha}\right]^{\alpha-1}K_l(\bs_2)^{1/\alpha}} \nonumber\\
&= {\sum_{l= 1}^LK_l(\bs_1)^{1/\alpha} + K_l(\bs_2)^{1/\alpha}\over \sum_{l= 1}^L K_l(\bs_2)^{1/\alpha}}\nonumber\\
&= 1 + {\sum_{l= 1}^LK_l(\bs_1)^{1/\alpha}\over \sum_{l= 1}^L K_l(\bs_2)^{1/\alpha}}.\label{eq:secondpart}
\end{align}}
Combining \eqref{eq:firstpart} and \eqref{eq:secondpart} together in (\ref{eq1}) gives $\chi(\bs_1, \bs_2)|K_l(\bs)= 0$ for all $\bs_1, \bs_2$ when $\theta >0$.

Finally, writing $\chi_z(\bs_1, \bs_2)|K_l(\bs)  =  \Pr\{Z(\bs_1) >z| Z(\bs_2)>z, K_l(\bs)\}$, we obtain
\begin{equation*}
\chi(\bs_1, \bs_2) = \underset{z \rightarrow \infty}{\lim} \chi_z(\bs_1, \bs_2) = \underset{z \rightarrow \infty}{\lim} {\rm E}\{\chi_z(\bs_1, \bs_2)|K_l(\bs)\}= {\rm E}\{\underset{z \rightarrow \infty}{\lim}  \chi_z(\bs_1, \bs_2)|K_l(\bs)\}= 0
\end{equation*}
for all $\bs_1, \bs_2 \in \mathcal{S}$, where the second to last line follows from the Dominated Convergence Theorem since $|\chi_z(\bs_1, \bs_2)|K_l(\bs)| < 1$.

\section{SPECTRAL REPRESENTATION} 
\label{sec:spectral_representation}
The finite-dimensional distributions of $\{Z(\bs), \bs\in \mathcal{S}\}$ defined as in (9) 
of the main text with $A_1,\ldots,A_L\iid {\rm H}(\alpha,\alpha,\theta)$, $\alpha\in(0,1)$, $\theta\geq0$, have the following spectral representation: let $\mathbf{X}_i=(X_{i,1}, \ldots, X_{i,D})^\top\in [0,\infty]^D\setminus\{\mathbf{0}\}$, $i=1,2,\ldots$, be points from a Poisson process with mean measure defined as
\begin{equation}\label{PPPmeasure}
\Lambda\left([0,\infty]^D\setminus[0,x_1]\times\ldots\times [0,x_D]\right) = \sum_{l = 1}^L\left[\theta + \sum_{j = 1}^D\{{x_j/K_l(\bs_j)}\}^{-1/\alpha}\right]^\alpha - L \theta^\alpha \geq 0.
\end{equation}
We note that the right-hand side of \eqref{PPPmeasure} is non-negative and monotone decreasing in each argument $x_j$, $j=1,\ldots,D$; therefore, $\Lambda$, endowed with the sigma-algebra of Borel sets, defines a valid Radon measure on $[0,\infty]^D\setminus\{\mathbf{0}\}$. Then, we construct the vector $\mathbf{Z}=\{Z(\bs_1),\ldots,Z(\bs_D)\}^\top$ as
\begin{equation}\label{PPP}
\mathbf{Z}=\max(\mathbf{0},\max_{i=1,2,\ldots}\mathbf{X}_i)=\max_{i=1,2,\ldots}\mathbf{X}_i,
\end{equation}
where the maximum is taken componentwise. The right-most equality in \eqref{PPP} follows from the fact that $\Lambda$ is infinite, i.e., $\Lambda\left([0,\infty]^D\setminus\{\mathbf{0}\}\right)=\infty$, which implies that an infinite number of Poisson points are sampled in $[0,\infty]^D\setminus\{\mathbf{0}\}$, and that the lower boundary $\mathbf{0}$ in \eqref{PPP} does not contribute to the maximum. From the spectral representation \eqref{PPP}, combined with \eqref{PPPmeasure}, we can check that we indeed recover the expression for our max-id model in (9) of the main text: 
\begin{align*}
\pr\left\{Z(\bs_1) \leq z_1, \ldots, Z(\bs_D)\leq z_D\right\} &= \pr\left(X_{i,1} \leq z_1, \ldots, X_{i,D}\leq z_D, \, i=1,2,\ldots\right) \\
&=\pr\left(\text{No point } \mathbf{X}_i \text{ in } [0,\infty]^D\setminus[0,z_1]\times,\ldots, \times [0,z_D] \right) \\
&= \exp\left(-\Lambda\{[0,\infty]^D\setminus[0,z_1]\times\ldots\times [0,z_D]\}\right)\\
&= \exp\left(L\theta^\alpha-\sum_{l = 1}^L\left[\theta + \sum_{j = 1}^D\{{z_j/K_l(\bs_j)}\}^{-1/\alpha}\right]^\alpha\right).
\end{align*}
The spectral representation \eqref{PPP} confirms that our model is indeed conditionally max-id, given the basis functions; see \citet{Resnick87}, Chapter 5, for more details on the spectral characterization of max-id random vectors. Such a characterization was recently exploited by \citet{Huser18} to construct alternative max-id models, which are quite flexible but much more intensive to fit than our proposed model.

\section{SIMULATING HOUGAARD RANDOM VARIABLES}
\label{sec:simulating_hougaard_random_variables}
Simulating from an exponentially tilted density using a simple rejection sampler becomes increasingly difficult as the exponential tilting parameter grows. The expected number of iterations for the following simple rejection sampler grows according to $1/{\rm E}\{\exp(-\theta X_1)\}$ if $X_1$ has density $f$ and $X_2$ follows an exponentially tilted $f$ with tilting parameter $\theta > 0$:
\begin{algorithmic}[1]
\Repeat 
\State Generate $X_1 \sim f$
\State Generate $U \sim {\rm Unif}(0,1)$
\Until{$U \leq \exp(-\theta X_1)$}
\State Set $X_2 = X_1$.
\end{algorithmic}
\cite{Devroye09} developed a fast double rejection method that is uniformly fast over all tilting parameters, which is useful when $\theta$ is large. For the proposed model, the more interesting values of $\theta$ occur when $\theta$ is small, for which the simple rejection sampler is sufficient. 

\section{MCMC DETAILS}
\label{sec:mcmc_details}

We implement Metropolis-Hastings MCMC algorithms in R (\url{http://www.r-project.org}) to draw posterior samples for the models described in this paper. Parameters $\alpha$, $\theta$, $\tau$, $\beta_\psi$, $\delta_\psi^2$, $\rho_\psi,$ $\psi \in \{\mu, \sigma\}$, $\xi$, $\delta_K^2$, and $\rho_K$ are all updated using variable-at-a-time Normal random walks. Scaling factors $A_{l,t}$ are updated using variable-at-a-time Normal random walks on a log-scale. Since we assume independence of the processes in different years, these basis scaling factors for different years may be updated in parallel. When $\theta > 0$, the density for $A_{l,t}$ can be expressed in terms of a $\mathrm{PS}(\alpha)$ density. The density function for the positive stable distribution is not available in closed form, but has a convenient integral representation \citep{Stephenson09,Kanter75,Chernin59}, which we evaluate numerically. For $\alpha \in (0,1)$, 
\[
f_{\rm PS}(x) = \int_0^1 {\alpha\over 1-\alpha} x^{-1/(1-\alpha)}a(\pi u)\exp\left\{-x^{-\alpha/(1-\alpha)}a(\pi u)\right\} {\rm d}u, 
\]
where
\[
a(v) = \left\{{\sin(\alpha v)\over\sin(v)}\right\}^{1/(1-\alpha)} {\sin\{(1-\alpha)v\}\over\sin(\alpha v)}.
\]

The computational burden associated with the log-Gaussian process basis functions and GEV marginal parameter Gaussian processes is one of the limiting factors in the scaling of the proposed model to many spatial locations. Often the strong correlations in the posterior distribution of a Gaussian process at nearby spatial locations make finding an efficient proposal difficult. To address this, we adapt a common sampling scheme that uses a Cholesky decomposition of the covariance matrix to transform the highly correlated Gaussian process to an i.i.d.\ scale \citep{Kuss05}. If $X$ is multivariate normal, $X\sim N(\mathbf{0}, \Sigma)$, then for Cholesky decomposition $\Sigma = LL^\top$, setting $Y = L^{-1}X$ gives $Y\sim N(\mathbf{0}, I)$. Performing Cholesky decompositions on large matrices at every MCMC iteration can quickly become expensive and is needless for our purposes if the range of the spatial dependence is not too large. With this in mind, we first partition our spatial region into disjoint sub-regions and perform block updates within each.  

We illustrate our proposal scheme for the GEV location parameter, but the same approach applies generically to others. Without loss of generality, assume that $\{\mu(\bs), \bs \in \mathbb{R}^2\}$ has mean zero Gaussian process prior with exponential covariance function $C(h;\rho_\mu, \delta^2_\mu)$, of which we would like to draw posterior samples. We first apply a $k$-nearest-neighbors clustering algorithm on the set of observation locations $\bs_1,\ldots, \bs_D$ to partition them into $J$ disjoint clusters, which we fix for the remainder of the algorithm. Let $\zeta(j)$ denote the set of $D_j$ spatial indices for the $j^{\mbox{\footnotesize th}}$ cluster. For the precipitation analysis, we take $J = 20$, however in general this choice should depend on the practical range of the process. Also, note that clustering on a spatially varying covariate instead of just the spatial locations could improve the efficiency of the sampler further.

Let $\bmu^{(m)}_j = \left[\mu\{\bs_{\zeta(j)_1}\}^{(m)}\ldots, \mu\{\bs_{\zeta(j)_{D_j}}\}^{(m)}\right]^\top$ be the $m^{\mbox{\footnotesize th}}$ MCMC draw of the process at the $D_j$ observation locations of the $j^{\mbox{\footnotesize th}}$ partition, and $\Sigma^{(m)}_{\mu_j}$ the corresponding covariance matrix obtained by evaluating $C(h;\rho^{(m)}_\mu, {\delta^2}^{(m)}_\mu)$ at all pairs of observation locations in the $j^{\mbox{\footnotesize th}}$ cluster, and $L_{\mu_j}$ its Cholesky factor. For each cluster, we perform block random-walk updates as follows:

\begin{algorithmic}[1]
\State Set $\breve{\bmu}^{(m)}_j = L_{\mu_j}^{-1}\bmu^{(m)}_j$
\State Propose $\breve{\bmu}^*_j \sim N(\breve{\bmu}^{(m)}_j, \lambda_{\mu_j}^2 I)$
\State Set $\bmu^*_j = L_{\mu_j} \breve{\bmu}^*_j$
\State Retain $\bmu^*_j$ with probability $R = \min\{1, f_\mu(\bmu^*_j|\Psi^{(m)}, \mathbf{z})/f_\mu(\bmu^{(m)}_j|\Psi^{(m)}, \mathbf{z})\}$,  
\end{algorithmic}

where $f_\mu$ is the full conditional distribution for $\bmu$ given the data $\mathbf{z}$ and remaining parameters $\Psi$, and the proposal variance $\lambda_{\mu_j}^2$ can be adaptively tuned \citep{Shaby10}.

The proposals for the Gaussian process basis functions are similar, but the blocking scheme also takes advantage of the dependence induced by the sum-to-one constraint. Let $\mathbf{K}^{(m)}_{l,j} = [K^{(m)}_l\{\bs_{\zeta(j)_1}\}, \ldots, \allowbreak K^{(m)}_l\{\bs_{\zeta(j)_{D_j}}\}]^\top$, $l =1, \ldots, L$, be the $m^{\mbox{\footnotesize th}}$ MCMC draw of the log-Gaussian process basis functions. For each spatial location $\bs$, we transform the basis functions to the Gaussian process scale $\tilde{K}^{(m)}_l(\bs) = \log\left\{K^{(m)}_l(\bs)/K^{(m)}_L(\bs)\right\}, \, l = 1, \ldots, L$. Then, just as before denoting the prior Gaussian process covariance matrix for the $j^{\mbox{\footnotesize th}}$ partition by the $\Sigma_{\tilde{K}_j}$ and corresponding Cholesky factor by $L_{\tilde{K}_j}$ (common across all $l = 1, \ldots, L-1$), the Metropolis update is
\begin{algorithmic}[1]
\For{$l = 1, \ldots, L-1$}
	\State Set $\breve{\mathbf{K}}^{(m)}_{l,j} = L_{\tilde{K}_j}^{-1}\tilde{\mathbf{K}}^{(m)}_{l,j}$ 
	\State Propose $\breve{\mathbf{K}}^{*}_{l,j} \sim N(\breve{\mathbf{K}}^{(m)}_{l,j}, \lambda_{K_j}^2 I)$
	\State Set $\tilde{\mathbf{K}}^{*}_{l,j} = L_{\tilde{K}_j} \breve{\mathbf{K}}^{*}_{l,j}$
\EndFor
\State Define $K^{*}_L(\bs) = 1$ for all $\bs$
\State Set $K^{*}_l(\bs_{\zeta(j)_i}) = \exp\left\{\tilde{K}^{*}_l(\bs_{\zeta(j)_i})\right\}/\sum_{l = 1}^L \exp\left\{\tilde{K}^{*}_l(\bs_{\zeta(j)_i})\right\}$, for $l = 1, \ldots, L$, and $i = 1, \ldots, D_j$ 
\State Retain $\mathbf{K}^{*}_{l,j}$ with probability $R = \min\{1, f_K(\mathbf{K}^{*}_{l,j}|\Psi^{(m)}, \mathbf{z})/f_K(\mathbf{K}^{(m)}_{l,j}|\Psi^{(m)}, \mathbf{z})\}$,  
\end{algorithmic}

where $f_K$ is the full conditional distribution for $\mathbf{K}^{*}_{l,j} = (K^{*}_l(\bs_{\zeta(j)_1}), \ldots, K^{*}_l(\bs_{\zeta(j)_{D_j}})^\top$, $l =1, \ldots, L$.

\section{POSTERIOR PREDICTIVE DRAWS}
\label{sec:posterior_predictive_draws}
In this section we describe how to make posterior predictive draws from the proposed model. We will make use of the following fact about multivariate normal distributions: let $\mathbf{W} = (\mathbf{W}_1^\top, \mathbf{W}_2^\top)^\top$, where $\mathbf{W}_1$ and $\mathbf{W}_2$ are $p$- and $q$-dimensional vectors, respectively. If $\mathbf{W}$ has multivariate normal distribution 

\[
\mathbf{W} \sim N\left(\begin{pmatrix}
\bmu_1 \\
\bmu_2
\end{pmatrix}, 
\begin{pmatrix}
\Sigma_{11} & \Sigma_{12} \\
\Sigma_{21} & \Sigma_{22} \\
\end{pmatrix}
\right),
\]
then $\mathbf{W}_1|\mathbf{W}_2 = \mathbf{w}_2 \sim N(\bmu_1 + \Sigma_{12}\Sigma^{-1}_{22}(\mathbf{w}_2 - \bmu_2), \Sigma_{11}- \Sigma_{12}\Sigma_{22}^{-1}\Sigma_{21})$. 

Now, let $m$ index the MCMC iteration, and let $\bs^*$ denote a prediction location. For each of $\bs_1,\ldots, \bs_D$, we first transform the basis functions to the Gaussian process scale by $\tilde{K}^{(m)}_l(\bs) = \log\left\{K^{(m)}_l(\bs)/K^{(m)}_L(\bs)\right\}, \, l = 1, \ldots, L$. Then, for $l = 1, \ldots, L-1$, we simulate the $l^{\mbox{\footnotesize th}}$ basis function $\tilde{K}^{(m)}_l(\bs^*)|\tilde{K}^{(m)}_l(\bs_1),\ldots, \tilde{K}^{(m)}_l(\bs_D), {\delta^{2}}^{(m)}, \rho^{(m)}$ from the corresponding conditional multivariate normal distribution. We set $\tilde{K}^{(m)}_L(\bs^*) = 1$, and normalize each basis function at the prediction location by $K_l^{(m)}(\bs^*) = \exp\{\tilde{K}_l^{(m)}(\bs^*)\}/\sum_{l = 1}^L \exp\{\tilde{K}_l^{(m)}(\bs^*)\}$, $l = 1, \ldots, L$. Then, we (conditionally) simulate GEV marginal parameters at prediction location from the corresponding multivariate normal distributions (e.g., $\mu^{(m)}(\bs^*)|\mu^{(m)}(\bs_1), \ldots, \mu^{(m)}(\bs_D), {\delta^{2}}_\mu^{(m)}, \rho_\mu^{(m)}$). Finally, we apply the following procedure to make draws of $\tilde{Z}^{(m)}_t(\bs^*)$ from the posterior predictive distribution, for each year $t$,

\begin{algorithmic}[1]
\State Set $Y^{(m)}_t(\bs^*)= \left\{\sum_{l=1}^L A^{(m)}_{l,t} \tilde{K}_l^{(m)}(\bs^*)^{1/\alpha}\right\}^{\alpha^{(m)}}$;
\State Draw $Z^{(m)}_t(\bs^*)|Y^{(m)}_t(\bs^*), \alpha^{(m)} \sim {\rm GEV}\{Y^{(m)}_t(\bs^*),\alpha^{(m)} Y^{(m)}_t(\bs^*) ,\alpha^{(m)}\}$;
\State Set $\tilde{Z}^{(m)}_t(\bs^*) = \mathrm{GEV}^{-1}\left[G_{\bs^*}\left\{Z^{(m)}_t(\bs^*)\right\};\mu^{(m)}(\bs^*), \sigma^{(m)}(\bs^*),\xi^{(m)}(\bs^*) \right]$,
\end{algorithmic}
where ${\rm GEV}^{-1}$ is the ${\rm GEV}$ quantile function, and $G_{\bs^*}$ denotes the (marginal) distribution function of $\tilde{Z}^{(m)}(\bs^*)$. 

\section{SIMULATION STUDY}
\label{sec:simulation_study}

\subsection{Simulation designs and results} 
\label{sub:designs_results}
To confirm that our MCMC algorithm produces reliable results, and to evaluate the algorithm's ability to infer the parameters of the process under different regimes, we conduct a simulation study for both the Gaussian density basis and the log-Gaussian process basis models. In all designs considered, we simulate $T = 30$ i.i.d.\ replicates (e.g., years) of the max-id process $\tilde{Z}(\bs)$ (defined in Section 2.2 of the main text) observed at $N = 100$ locations uniformly distributed on the unit square $[0,1] \times [0,1]$. In this section only, we replace the Gaussian process priors on the GEV parameters with simpler non-spatially varying priors, as the computational burden is already quite high without spatially varying margins, and instead we use $N(0,100)$ priors for $\mu, \gamma \equiv \log(\sigma),$ and $\xi$. For all simulations, data are generated using $\mu = 0, \, \sigma = 1,$ and $\xi = 0$ (i.e., with standard Gumbel margins). For the Gaussian density basis model, we consider $L = 25$ knots evenly distributed over the unit square using a standard deviation of $\tau = 1/6$, and for the log-Gaussian process basis model, we use $L = 15$ basis functions and take the variance and range parameters of the underlying Gaussian processes to be $\delta_K^2 = 25$ and $\rho_K = 3/4$. We vary the settings of the dependence parameters $\alpha$ and $\theta$ as described in Table \ref{tab:simset}. For each simulation design, models are fitted to $100$ datasets using MCMC (details are given in Appendix \ref{sec:mcmc_details}). The effective sample sizes (ESS) and effective samples per second (ES/sec) for each design are reported in Section \ref{sub:computation_time}.

\begin{table}[t!]
\begin{center}
\caption{Dependence parameter simulation settings used for the Gaussian density and log-Gaussian process models.}
\label{tab:simset}
\centering
\begin{tabular}{r|rrrrrr}
 Sim. \#    &   1 & 2 & 3 & 4 & 5 & 6 \tabularnewline
 \hline
$\alpha$&$ 0.1$&$ 0.25$&$ 0.1$&$ 0.25$&$ 0.1$&$ 0.25$\tabularnewline
$\theta$&$ 0$&$ 0$&$ 10^{-4}$&$ 10^{-4}$&$ 0.1$&$ 0.1$\tabularnewline
\end{tabular}
\end{center}
\end{table}

The coverage of Bayesian $95\%$ highest posterior density (HPD) credible intervals, bias, and root mean square error (RMSE) of the posterior mean estimates for the Gaussian density and log-Gaussian process basis models are reported in Tables \ref{tab:gden_sim_res} and \ref{tab:lgproc_sim_res}, respectively. The simulation results show nearly nominal coverages in all cases. The slight positive bias of $\delta_K^2$ and $\rho_K$ relative to their magnitudes, but nominal coverage of the 95\% credible intervals, suggests that these two parameters may not be completely identifiable in practice. However, since we are more concerned with making inference on the quantiles and general spatial patterns of extremes, this is not a major concern.

\begin{table}[t!]
\begin{center}
\caption{Bias, root mean squared error, and $95\%$ credible interval coverages for parameters of Gaussian density basis models.}
\label{tab:gden_sim_res}
\begin{tabular}{r|rrrrrr}
\toprule
\multicolumn{1}{r}{}&\multicolumn{1}{c}{$\alpha$}&\multicolumn{1}{c}{$\theta$}&\multicolumn{1}{c}{$\tau$}&\multicolumn{1}{c}{$\mu$}&\multicolumn{1}{c}{$\sigma$}&\multicolumn{1}{c}{$\xi$}\tabularnewline
\midrule
\midrule
{\bfseries Bias}&&&&&&\tabularnewline
Sim. 1&$  0.0005$&$     		$&$ -0.00006 $&$ -0.006   $&$ -0.0009  $&$   0.000003$\tabularnewline
Sim. 2&$  0.002   $&$   		$&$ -0.0003	 $&$ -0.003   $&$ -0.005   $&$   0.0002  $\tabularnewline
Sim. 3&$  0.0003  $&$   0.0001  $&$  0.0009  $&$ -0.005   $&$  0.001   $&$  -0.002   $\tabularnewline
Sim. 4&$ -0.003   $&$   0.00003 $&$  0.002   $&$  0.004   $&$  0.003   $&$  -0.0002  $\tabularnewline
Sim. 5&$  0.004   $&$   0.001   $&$ -0.002   $&$ -0.01   $&$  0.0005  $&$   0.003   $\tabularnewline
Sim. 6&$  0.01  $&$   0.02   $&$ -0.002   $&$ -0.001   $&$ -0.001   $&$   0.002   $\tabularnewline
\midrule
{\bfseries RMSE}&&&&&&\tabularnewline
Sim. 1&$  0.002   $&$           $&$  0.001   $&$  0.06    $&$  0.02    $&$   0.002   $\tabularnewline
Sim. 2&$  0.009   $&$           $&$  0.003   $&$  0.07    $&$  0.03    $&$   0.006   $\tabularnewline
Sim. 3&$  0.007   $&$   0.0003  $&$  0.006   $&$  0.04    $&$  0.02    $&$   0.01    $\tabularnewline
Sim. 4&$  0.02    $&$   0.00007 $&$  0.009   $&$  0.06    $&$  0.03    $&$   0.02    $\tabularnewline
Sim. 5&$  0.007   $&$   0.05    $&$  0.005   $&$  0.04    $&$  0.01    $&$   0.01    $\tabularnewline
Sim. 6&$  0.03    $&$   0.06    $&$  0.01    $&$  0.05    $&$  0.02    $&$   0.01    $\tabularnewline
\midrule
{\bfseries Coverage}&&&&&&\tabularnewline
Sim. 1&$  0.93$&$       $&$  0.97$&$  0.91$&$  0.96$&$   0.94$\tabularnewline
Sim. 2&$  0.97$&$       $&$  0.94$&$  0.96$&$  0.97$&$   0.94$\tabularnewline
Sim. 3&$  0.95$&$   0.93$&$  0.94$&$  0.94$&$  0.94$&$   0.96$\tabularnewline
Sim. 4&$  0.91$&$   0.96$&$  0.95$&$  0.94$&$  0.92$&$   0.91$\tabularnewline
Sim. 5&$  0.96$&$   0.89$&$  0.98$&$  0.97$&$  0.96$&$   0.94$\tabularnewline
Sim. 6&$  0.94$&$   0.94$&$  0.94$&$  0.93$&$  0.96$&$   0.96$\tabularnewline
\bottomrule
\end{tabular}\end{center}
\end{table}

\begin{table}[t!]
\begin{center}
\caption{Bias, root mean squared error, and $95\%$ credible interval coverages for parameters of log-Gaussian process basis models.}
\label{tab:lgproc_sim_res}
\begin{tabular}{r|rrrrrrr}
\toprule
\multicolumn{1}{r}{}&\multicolumn{1}{c}{$\alpha$}&\multicolumn{1}{c}{$\theta$}&\multicolumn{1}{c}{$\delta_K^2$}&\multicolumn{1}{c}{$\rho_K$}&\multicolumn{1}{c}{$\mu$}&\multicolumn{1}{c}{$\sigma$}&\multicolumn{1}{c}{$\xi$}\tabularnewline
\midrule
{\bfseries Bias}&&&&&&&\tabularnewline
Sim. 1&$ -0.00003  $&$            $&$  2.2      $&$  0.06     $&$  0.006    $&$ -0.0004   $&$  -0.0003   $\tabularnewline
Sim. 2&$  0.002    $&$            $&$  3.7      $&$  0.09     $&$ -0.0004   $&$ -0.002    $&$   0.0007   $\tabularnewline
Sim. 3&$  0.0009   $&$   0.0001   $&$  2.7      $&$  0.08     $&$ -0.002    $&$ -0.001    $&$  -0.001    $\tabularnewline
Sim. 4&$  0.002    $&$   0.00002  $&$  4.5      $&$  0.07     $&$  0.007    $&$  0.0004   $&$   0.002    $\tabularnewline
Sim. 5&$  0.001    $&$   0.005    $&$  2.7      $&$  0.06     $&$ -0.02     $&$ -0.0005   $&$   0.002    $\tabularnewline
Sim. 6&$  0.003    $&$   0.03     $&$  5.8      $&$  0.07     $&$ -0.005    $&$ -0.003    $&$   0.0006   $\tabularnewline
\midrule
{\bfseries RMSE}&&&&&&&\tabularnewline
Sim. 1&$  0.002$&$            $&$  4.5$&$  0.14	$&$  0.06 $&$  0.007$&$   0.004$\tabularnewline
Sim. 2&$  0.006$&$            $&$  5.8$&$  0.15	$&$  0.08 $&$  0.02 $&$   0.008$\tabularnewline
Sim. 3&$  0.004$&$   0.0003   $&$  5.1$&$  0.16	$&$  0.07 $&$  0.01 $&$   0.01$\tabularnewline
Sim. 4&$  0.009$&$   0.00007  $&$  7.2$&$  0.14	$&$  0.08 $&$  0.02 $&$   0.02$\tabularnewline
Sim. 5&$  0.005$&$   0.06     $&$  5.0$&$  0.14	$&$  0.07 $&$  0.02 $&$   0.01$\tabularnewline
Sim. 6&$  0.02$&$   0.06      $&$  8.7$&$  0.15	$&$  0.06 $&$  0.03 $&$   0.02$\tabularnewline
\midrule
{\bfseries Coverage}&&&&&&&\tabularnewline
Sim. 1&$  0.91$&$        $&$  0.95$&$  0.91$&$  0.93$&$  0.97$&$   0.91$\tabularnewline
Sim. 2&$  0.93$&$        $&$  0.93$&$  0.92$&$  0.95$&$  0.98$&$   0.96$\tabularnewline
Sim. 3&$  0.96$&$   0.94$&$  0.96$&$  0.89$&$  0.86$&$  0.94$&$   0.92$\tabularnewline
Sim. 4&$  0.95$&$   0.94$&$  0.97$&$  0.99$&$  0.93$&$  0.93$&$   0.95$\tabularnewline
Sim. 5&$  0.94$&$   0.93$&$  0.97$&$  0.94$&$  0.90$&$  0.94$&$   0.96$\tabularnewline
Sim. 6&$  0.93$&$   0.93$&$  0.92$&$  0.93$&$  0.95$&$  0.95$&$   0.95$\tabularnewline
\bottomrule
\end{tabular}\end{center}
\end{table}

\subsection{Computation time} 
\label{sub:computation_time}
The average effective sample sizes (ESS), and effective samples per second (ES/sec) for the simulations are reported in Tables \ref{tab:lgproc_comp_time} and \ref{tab:gden_comp_time}. In general, the efficiency of samplers for the log-Gaussian process basis models is better than for the Gaussian density basis models due to the smaller number of $A_{l,t}$ terms. Also, samplers for the max-stable ($\theta=0$) models are more efficient than their max-id counterparts due to the closed-form expression of the marginal quantile functions for $\tilde{Z}(\bs)$ in the max-stable case. For the max-id models, the marginal quantile functions for $\tilde{Z}(\bs)$ are obtained by numerical inversion of the distribution function.
\begin{table}[t!]
\begin{center}
\caption{The average effective sample sizes (ESS), and effective samples per second (ES/sec) for MCMC samplers of the log-Gaussian process basis simulation study.}
\label{tab:lgproc_comp_time}
\begin{tabular}{r|rrrrrrr}
\toprule
\multicolumn{1}{r}{}&\multicolumn{1}{c}{$\alpha$}&\multicolumn{1}{c}{$\theta$}&\multicolumn{1}{c}{$\delta^2_K$}&\multicolumn{1}{c}{$\rho_K$}&\multicolumn{1}{c}{$\mu$}&\multicolumn{1}{c}{$\sigma$}&\multicolumn{1}{c}{$\xi$}\tabularnewline
\midrule
{\bfseries ESS}&&&&&&&\tabularnewline
Sim. 1&$123$&$        $&$379$&$378$&$109$&$227$&$ 408$\tabularnewline
Sim. 2&$172$&$        $&$336$&$336$&$128$&$214$&$ 583$\tabularnewline
Sim. 3&$134$&$ 118$&$366$&$365$&$115$&$413$&$ 803$\tabularnewline
Sim. 4&$218$&$ 185$&$281$&$291$&$137$&$368$&$ 900$\tabularnewline
Sim. 5&$121$&$ 116$&$340$&$345$&$123$&$380$&$1080$\tabularnewline
Sim. 6&$136$&$ 128$&$243$&$262$&$187$&$434$&$1752$\tabularnewline
\midrule
{\bfseries ES/sec}&&&&&&&\tabularnewline
Sim. 1&$  0.004$&$        $&$  0.011$&$  0.011$&$  0.003$&$  0.007$&$   0.012$\tabularnewline
Sim. 2&$  0.005$&$        $&$  0.010$&$  0.010$&$  0.004$&$  0.007$&$   0.018$\tabularnewline
Sim. 3&$  0.002$&$   0.002$&$  0.005$&$  0.005$&$  0.002$&$  0.006$&$   0.011$\tabularnewline
Sim. 4&$  0.003$&$   0.003$&$  0.004$&$  0.005$&$  0.002$&$  0.006$&$   0.014$\tabularnewline
Sim. 5&$  0.002$&$   0.002$&$  0.005$&$  0.005$&$  0.002$&$  0.006$&$   0.017$\tabularnewline
Sim. 6&$  0.002$&$   0.002$&$  0.004$&$  0.004$&$  0.003$&$  0.007$&$   0.028$\tabularnewline
\bottomrule
\end{tabular}\end{center}
\end{table}

\begin{table}[t!]
\begin{center}
\caption{The average effective sample sizes (ESS), and effective samples per second (ES/sec) for MCMC samplers of the Gaussian density basis simulation study.}
\label{tab:gden_comp_time}
\begin{tabular}{r|rrrrrr}
\toprule
\multicolumn{1}{r}{}&\multicolumn{1}{c}{$\alpha$}&\multicolumn{1}{c}{$\theta$}&\multicolumn{1}{c}{$\tau$}&\multicolumn{1}{c}{$\mu$}&\multicolumn{1}{c}{$\sigma$}&\multicolumn{1}{c}{$\xi$}\tabularnewline
\midrule
{\bfseries ESS}&&&&&&\tabularnewline
Sim. 1&$111$&$    $&$138$&$111$&$135$&$1527$\tabularnewline
Sim. 2&$121$&$    $&$147$&$131$&$137$&$1562$\tabularnewline
Sim. 3&$128$&$ 114$&$127$&$145$&$387$&$1452$\tabularnewline
Sim. 4&$169$&$ 183$&$172$&$174$&$250$&$ 553$\tabularnewline
Sim. 5&$121$&$ 112$&$127$&$186$&$682$&$2859$\tabularnewline
Sim. 6&$128$&$ 119$&$138$&$446$&$661$&$3308$\tabularnewline
\midrule
{\bfseries ES/sec}&&&&&&\tabularnewline
Sim. 1&$  0.002$&$   0.166$&$  0.003$&$  0.002$&$  0.003$&$   0.032$\tabularnewline
Sim. 2&$  0.003$&$   0.177$&$  0.003$&$  0.003$&$  0.003$&$   0.034$\tabularnewline
Sim. 3&$  0.001$&$   0.001$&$  0.001$&$  0.001$&$  0.003$&$   0.013$\tabularnewline
Sim. 4&$  0.002$&$   0.002$&$  0.002$&$  0.002$&$  0.002$&$   0.005$\tabularnewline
Sim. 5&$  0.001$&$   0.001$&$  0.001$&$  0.002$&$  0.006$&$   0.026$\tabularnewline
Sim. 6&$  0.001$&$   0.001$&$  0.001$&$  0.004$&$  0.006$&$   0.032$\tabularnewline
\bottomrule
\end{tabular}\end{center}
\end{table}

\section{Additional Precipitation Analysis} 
\label{sec:additional_precipitation_analysis}
In Figure 6 of the main text, the QQ-plots of group-wise minima and maxima do not correspond perfectly. To assess the source of this slight discrepancy in the higher order dependence characteristics, we perform the same analysis of group-wise statistics at observation locations. The results (Figure \ref{fig:qq_obs}) show closer correspondence between the observed and model distributions of group-wise minima, mean, and maxima, which is expected. The consistency of the model based and empirical distributions of these group-wise statistics at both observation and holdout locations supports the models ability to capture the dependence in the annual maxima fields.

\begin{figure}[!htb]
  \begin{center}
		\includegraphics[width=\textwidth]{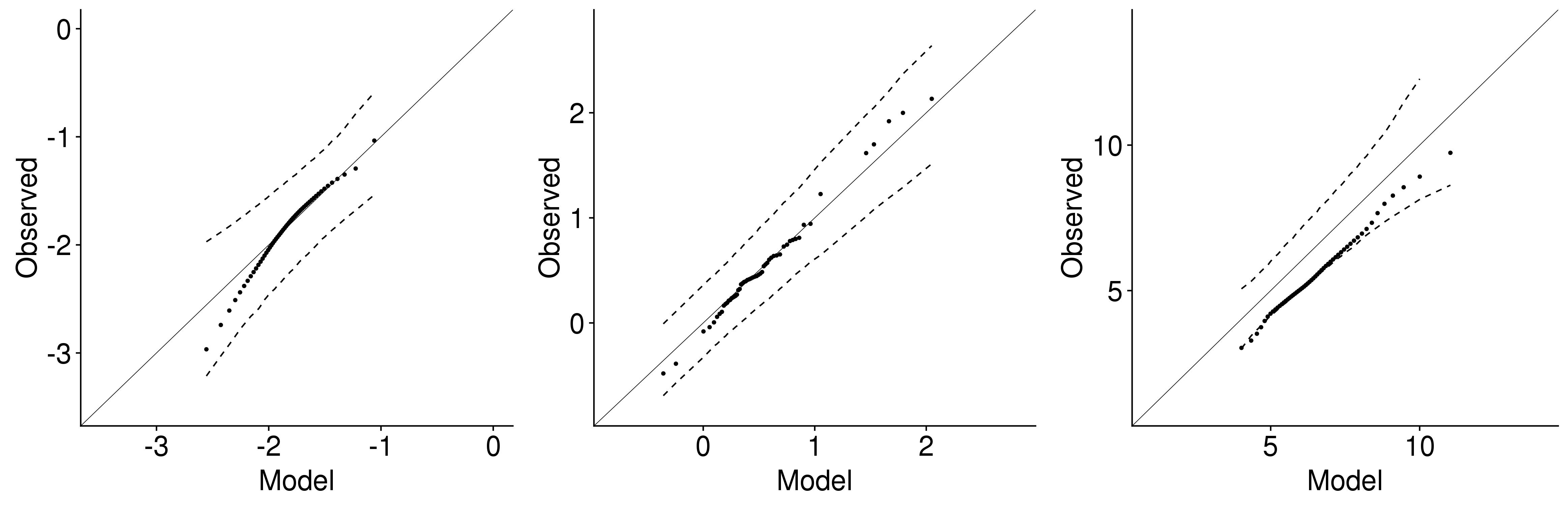}
		\caption{QQ-plots of the observed and predicted group-wise minima (left), mean (center), and maxima (right) taken over the annual maxima from all 546 observation stations. The dashed lines represent 95\% credible intervals. The plots reflect close correspondence between the empirical and modeled multivariate distributions. To account for the fact that the marginal GEV distributions vary across stations, observations are first transformed to unit Gumbel scale using the probability integral transform for the GEV marginal distributions at each station from the fitted model.}
		\label{fig:qq_obs}
	\end{center}
\end{figure}

Figure 5 of the main text shows only a slight discrepancy between the empirical and the max-id model-based estimates of $\chi_u$. For comparison, the log-Gaussian process basis, max-stable ($\theta = 0)$ model 95\% credible intervals are shown in Figure \ref{fig:lnms_precip_chi}. We see greater underestimation of dependence at short distances in the max-stable case than was observed in the max-id case, indicating that the source of the discrepancy diminished by the added flexibility in spatial dependence from the additional parameter $\theta > 0$. The discrepancy appears to be due to difficulty in accommodating the precipitation data's quite strong dependence at short lags that persists across quantiles and weak dependence at long lags that continues to quickly decay for increasing quantiles.

\begin{figure}[t!]
  \begin{center}
		\includegraphics[width=\textwidth]{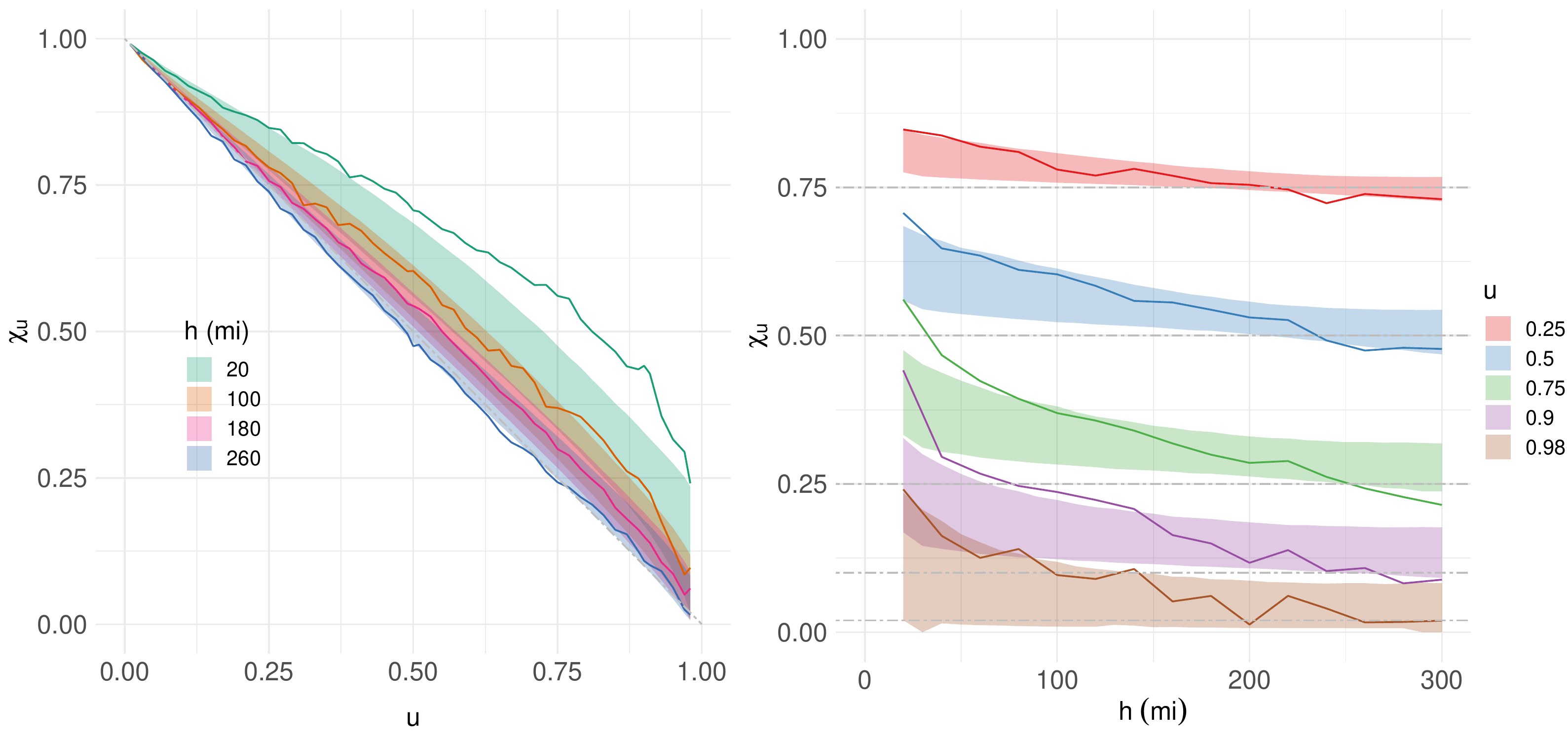}
		\caption{The left panel shows $\chi_u$ as a function of $u$ for fixed spatial lags $h = 20, 100, 180, 260$ miles calculated for the 100-holdout stations. Empirical estimates are shown as a solid black line, and max-stable, log-Gaussian Process basis model 95\% credible intervals are shown as gray ribbons. The decay of $\chi_u$ towards zero as $u \rightarrow 1$ suggests that daily precipitation are asymptotically independent. To understand the spatial dependence of extreme precipitation at increasingly extreme levels, empirical (solid lines) and model 95\% credible intervals (ribbons) of $\chi_u(h)$ for the holdout stations are plotted for several quantiles $u = 0.25, 0.5, 0.75, 0.9, 0.98$ (right panel). Horizontal dash-dot gray lines representing the values of $\chi_u$ under an everywhere-independent model are plotted for reference. The plot shows good overall agreement between the model fits and empirical estimates, except at very short distances.}
		\label{fig:lnms_precip_chi}
	\end{center}
\end{figure}

To assess the sensitivity of the posterior factor mean estimates to the choice of the number of basis functions $L$, here the top ranked factor means for other choices of $L$ are shown in analogous plots to those in Figure 9 of the main text. Several of the top ranked posterior factor means in the $L = 10$ and $L = 20$ log-Gaussian process basis function, $\theta > 0$ models agree with those identified in the $L = 15$ case. The top six ranked factor means for $L = 10$ and $L = 20$ are shown in Figures \ref{fig:L10_factors} and \ref{fig:L20_factors}. There are greater discrepancies between the $L = 10$ model and the others, which may be due to the rigidity imposed by using relatively few basis functions.

\begin{figure}[t!]
  \begin{center}
		\includegraphics[width=\textwidth]{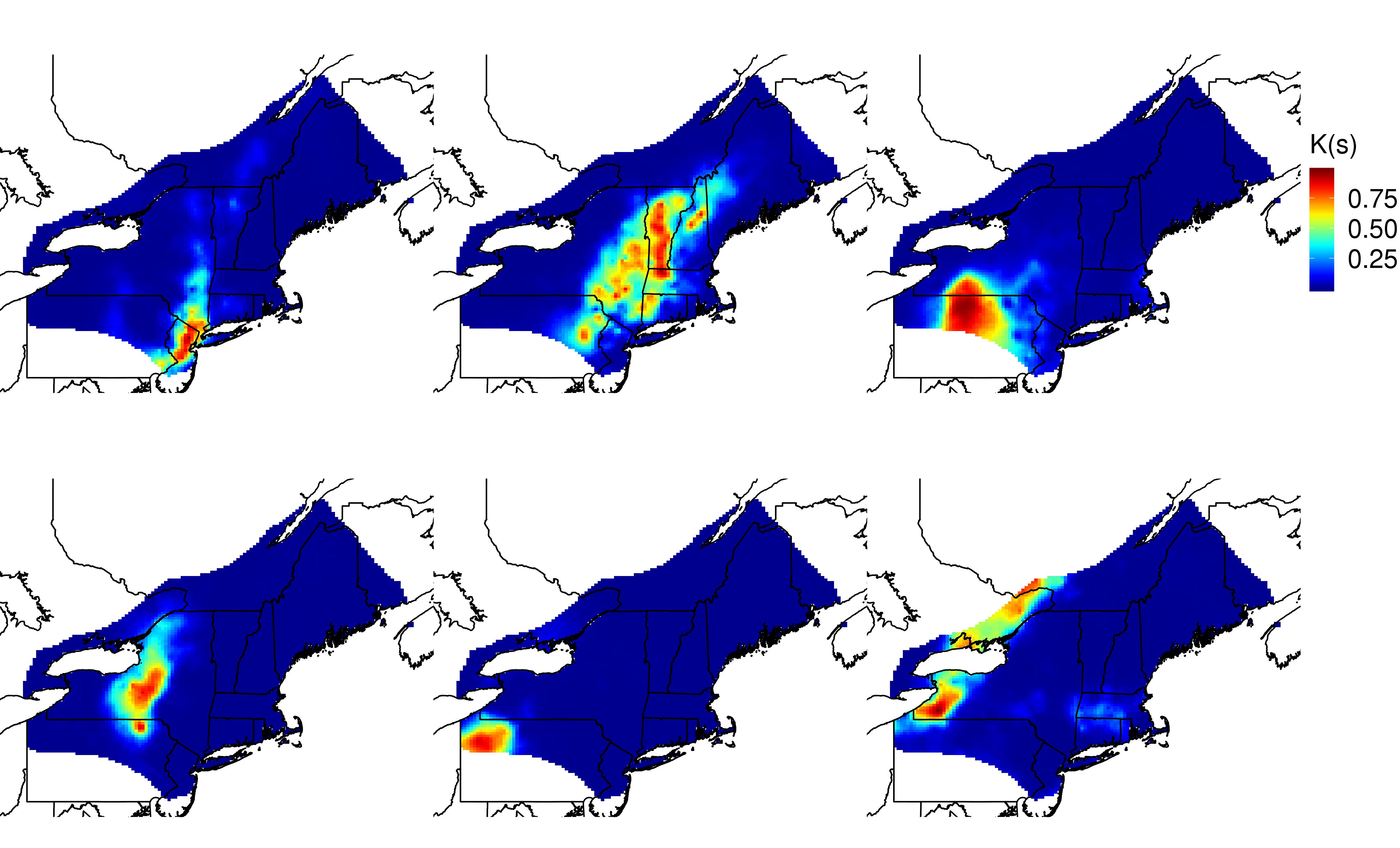}
		\caption{The top 6 ranked posterior factor means for the log-Gaussian process basis, $\theta > 0$ model with $L = 10$ basis functions. The (1) New-Jersey and (2) north-south New York-Vermont basis peaks are similar to those in the $L = 15$ model.}
		\label{fig:L10_factors}
	\end{center}
\end{figure}

\begin{figure}[t!]
  \begin{center}
		\includegraphics[width=\textwidth]{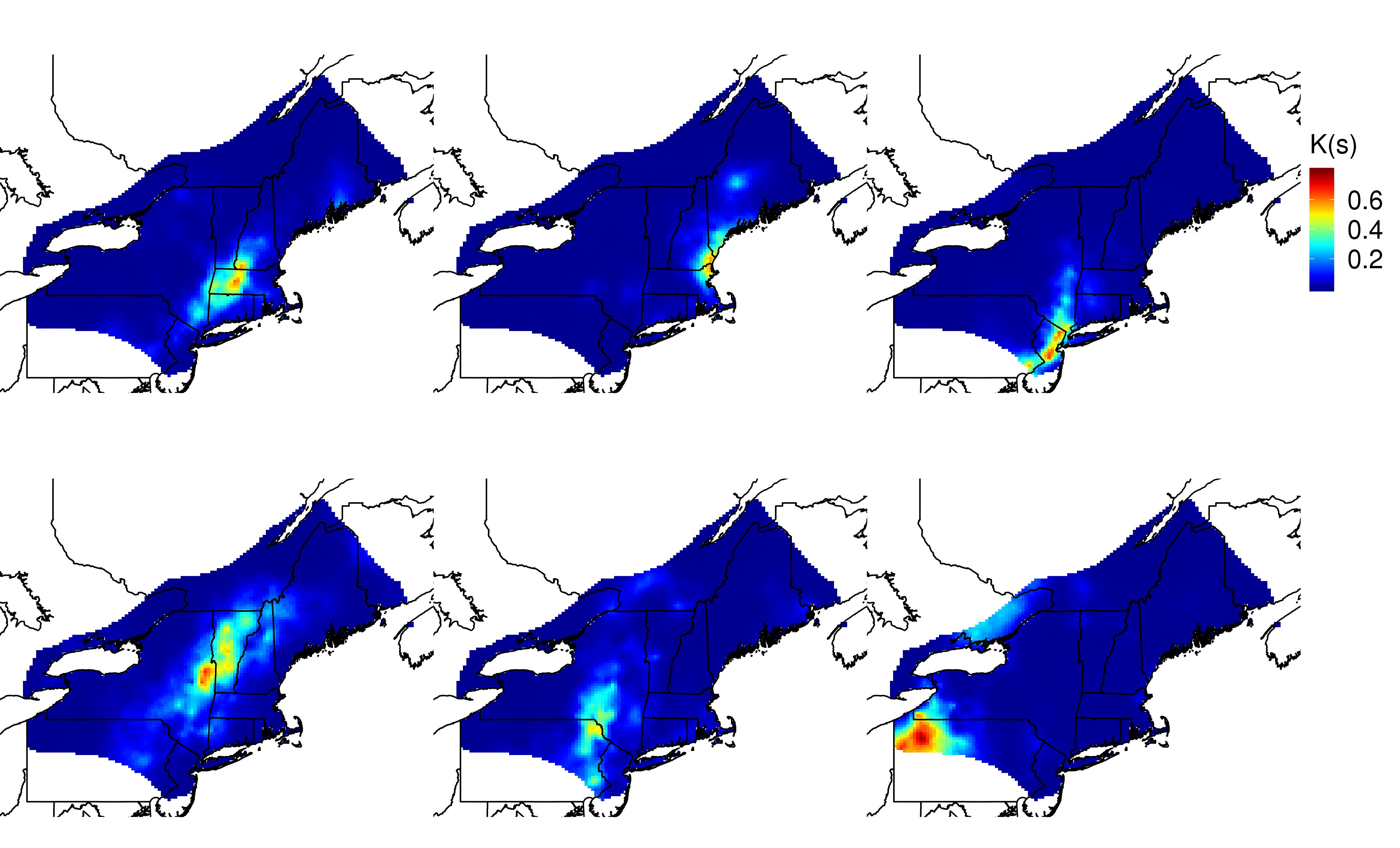}
		\caption{The top 6 ranked posterior factor means for the log-Gaussian process basis, $\theta > 0$ model with $L = 20$ basis functions. The (1) western Massachusetts, (2) eastern Massachusetts, (3) New Jersey, and (4) north-south New York-Vermont basis are similar to those in the $L = 15$ model.}
		\label{fig:L20_factors}
	\end{center}
\end{figure}
Next, we assess the appropriateness of the max-infinite divisibility assumption for the distribution of precipitation annual maxima. A distribution function $F$ on $\mathbb{R}^2$ is max-id if and only if for every rectangle $R = [a_0, a_1) \times [b_0, b_1)$,
\begin{equation}
\Delta \equiv F_{11}F_{00} - F_{01}F_{10} \geq 0,
\label{eq:bi_max_id_cond}
\end{equation}
where $F_{ij} = F(a_i,b_j)$ for $i,j\in \{0,1\}$ \citep{Balkema77}. We assess the empirical bivariate distributions of annual maximum precipitation for pairs of rain gauge stations. Consider a given pair of stations at locations $\bs_i$ and $\bs_j$, with empirical precipitation marginal distributions $\hat{F}^{(i)}$ and $\hat{F}^{(j)}$. We sample $M = 1,000$ rectangles as follows: we generate independent real numbers $\tilde{a}_0^m, \tilde{a}_1^m \iid \hat{F}^{(i)}$ and $\tilde{b}_0^m, \tilde{b}_1^m \iid \hat{F}^{(j)}$, with $m=1,\ldots,M$, and define the coordinates of the $m$-th rectangle as $(a_0^m, a_1^m)= (\min\{\tilde{a}_0^m, \tilde{a}_1^m\}, \max\{\tilde{a}_0^m, \tilde{a}_1^m\})$ and $(b_0^m, b_1^m)= (\min\{\tilde{b}_0^m, \tilde{b}_1^m\}, \max\{\tilde{b}_0^m, \tilde{b}_1^m\})$. For rectangle $R_m=[a_0^m, a_1^m) \times [b_0^m, b_1^m)$, we set $\Delta_m(\bs_i, \bs_j)$ equal to the difference from Equation \ref{eq:bi_max_id_cond}, and consider $\Delta_{\min}(\bs_i, \bs_j) = \min_{m = 1}^M \Delta_m(\bs_i, \bs_j)$ and $p_{\Delta>0}(\bs_i, \bs_j) = {1\over M}\sum_{m = 1}^M I(\Delta_m(\bs_i, \bs_j)>0)$, where $I(\cdot)$ is the indicator function. We repeat this experiment $N = 300$ times. Densities of the proportion $p_{\Delta>0}$ of rectangles for which $\Delta_m > 0$ and densities of the lower $5\%$ quantile of minimum differences $\Delta_{\min}$ are summarized in Figure \ref{fig:max_id_summary}  by distance between pairs of stations up to 100 miles apart.

For comparison, we repeat the experiment by sampling from the bivariate normal distribution, and bivariate extreme value distribution with logistic dependence structure. We draw samples of size $n = 55$ and calculate $M = 1,000$ rectangles for each of $N = 300$ Monte Carlo replicates of the experiment for varying dependence parameters. For the bivariate normal case, we take standard normal marginal distributions and varying the correlation parameter $\tilde{\rho}$ between $-0.9$ and $0.9$. For the bivariate extreme value distribution, we take unit Fr\'{e}chet marginal distributions, with varying logistic dependence parameter $\tilde{\alpha}$ between $0.05$ and $1$. All of the bivariate extreme value distributions considered are max-id. The bivariate normal distributions with $\tilde{\rho} < 0$ are not max-id, and it is not clear when $\tilde{\rho}>0$. The results are reported in Tables~\ref{tab:bivar_evd} (extreme-value dependence structure) and \ref{tab:bivar_gauss} (Gaussian dependence structure). The pairwise empirical distributions of precipitation annual maxima appear to be consistent with the max-id assumption, as Equation \ref{eq:bi_max_id_cond} is satisfied for the majority of rectangles across all pairwise station distances.

\begin{figure}[t!]
    \centering
    \begin{subfigure}[b]{0.45\textwidth}
        \includegraphics[width=\textwidth]{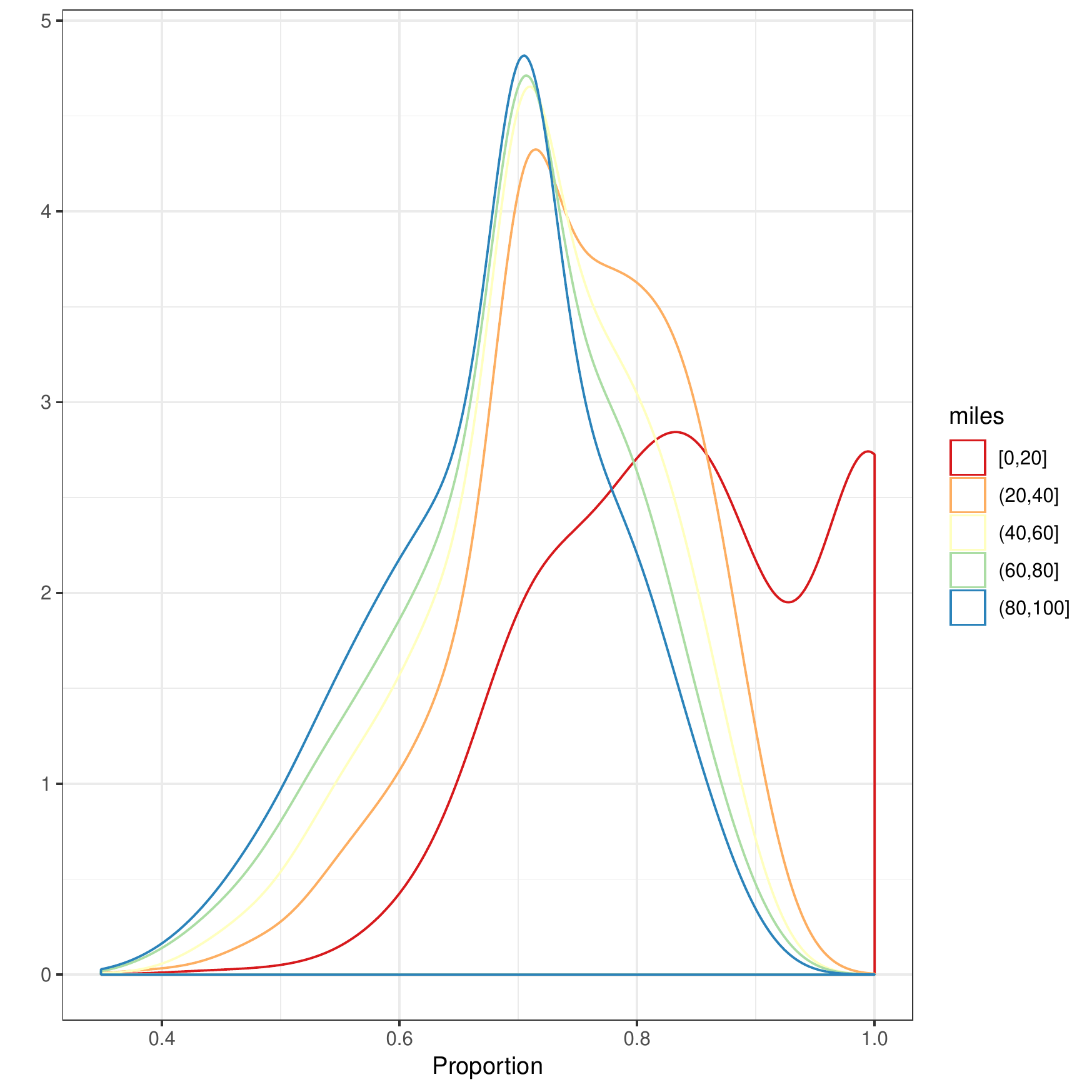}
    \end{subfigure}
    ~ 
    \begin{subfigure}[b]{0.45\textwidth}
        \includegraphics[width=\textwidth]{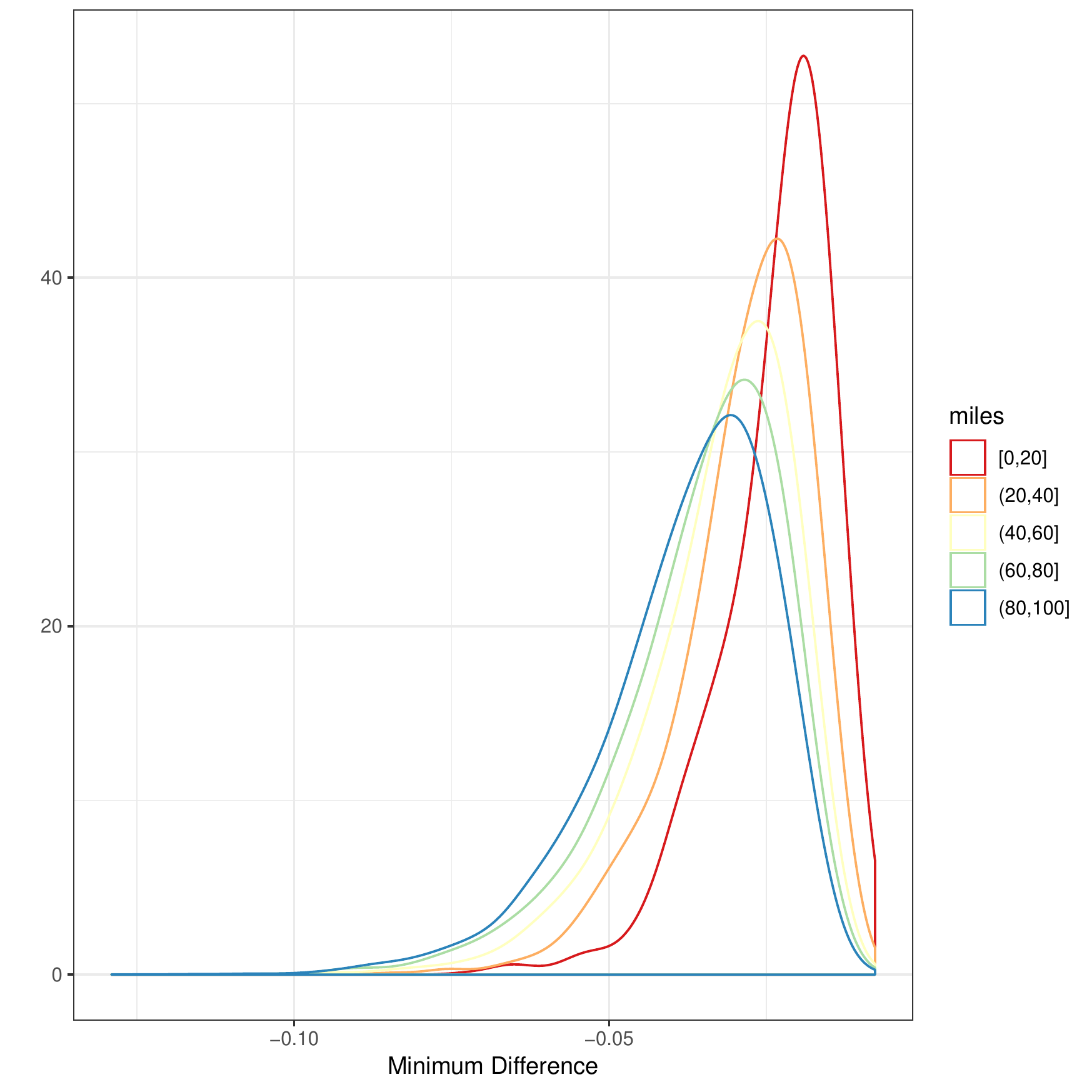}
    \end{subfigure}
    \caption{Densities of the proportion of rectangles $p_{\Delta>0}$ of rectangles for which $\Delta_m > 0$ (left) and densities of the lower $5\%$ quantile of minimum differences $\Delta_{\min}$ (right).}
    \label{fig:max_id_summary}
\end{figure}

\begin{table}[t!]
\centering
\caption{Summary of dependence for bivariate extreme value distribution with logistic dependence structure and dependence parameter $\tilde{\alpha}$. Monte Carlo estimates of the $2.5\%$ and $97.5\%$ quantiles of the proportion of rectangles for which $\Delta >0$, as well as the $5\%$ quantile of $\Delta_{\min}$ over $M = 1000$ rectangles, each column based on $N = 300$ Monte Carlo samples.}
\label{tab:bivar_evd}
\resizebox{\textwidth}{!}{
\begin{tabular}{r|r|rrrrrrrrrrrrrrrrrrrr}
  \hline
&$\tilde{\alpha}$ & 0.05 & 0.1 & 0.15 & 0.2 & 0.25 & 0.3 & 0.35 & 0.4 & 0.45 & 0.5 & 0.55 & 0.6 & 0.65 & 0.7 & 0.75 & 0.8 & 0.85 & 0.9 & 0.95 & 1 \\   \hline
Prop. $\Delta > 0$ &2.5\% & 0.99 & 0.98 & 0.96 & 0.95 & 0.91 & 0.88 & 0.84 & 0.80 & 0.72 & 0.72 & 0.69 & 0.62 & 0.55 & 0.53 & 0.48 & 0.47 & 0.45 & 0.43 & 0.39 & 0.40 \\ 
&  97.5\% & 1.00 & 1.00 & 1.00 & 0.99 & 0.99 & 0.98 & 0.98 & 0.96 & 0.96 & 0.94 & 0.93 & 0.91 & 0.90 & 0.89 & 0.87 & 0.86 & 0.84 & 0.82 & 0.77 & 0.76 \\ 
   \hline
  $\Delta_{\min}$& 5.0\% & -0.00 & -0.00 & -0.01 & -0.01 & -0.01 & -0.01 & -0.01 & -0.02 & -0.02 & -0.02 & -0.03 & -0.03 & -0.03 & -0.04 & -0.04 & -0.05 & -0.05 & -0.06 & -0.07 & -0.07 \\ 
   \hline
\end{tabular}
}
\end{table}

\begin{table}[t!]
\centering
\caption{Summary of dependence for bivariate normal with correlation $\tilde{\rho}$. Monte Carlo estimates of the $2.5\%$ and $97.5\%$ quantiles of the proportion of rectangles for which $\Delta >0$, as well as the $5\%$ quantile of $\Delta_{\min}$ over $M = 1000$ rectangles, each column based on $N = 300$ Monte Carlo samples.}
\label{tab:bivar_gauss}
\resizebox{\textwidth}{!}{
\begin{tabular}{r|r|rrrrrrrrrrrrrrrrrrr}
  \hline
 &$\tilde{\rho}$ & 0.9 & 0.8 & 0.7 & 0.6 & 0.5 & 0.4 & 0.3 & 0.2 & 0.1 & 0 & -0.1 & -0.2 & -0.3 & -0.4 & -0.5 & -0.6 & -0.7 & -0.8 & -0.9 \\ 
  \hline
 Prop. $\Delta > 0$ &2.5\% & 0.92 & 0.85 & 0.78 & 0.70 & 0.63 & 0.57 & 0.50 & 0.44 & 0.42 & 0.37 & 0.38 & 0.35 & 0.35 & 0.35 & 0.36 & 0.37 & 0.38 & 0.41 & 0.48 \\ 
  &97.5\% & 0.99 & 0.97 & 0.95 & 0.94 & 0.91 & 0.90 & 0.87 & 0.83 & 0.81 & 0.79 & 0.73 & 0.71 & 0.66 & 0.61 & 0.58 & 0.55 & 0.54 & 0.54 & 0.58 \\ 
  \hline
$\Delta_{\min}$&5.0\% & -0.01 & -0.02 & -0.02 & -0.03 & -0.03 & -0.03 & -0.04 & -0.05 & -0.07 & -0.07 & -0.07 & -0.10 & -0.10 & -0.11 & -0.12 & -0.14 & -0.15 & -0.17 & -0.19 \\ 
   \hline
\end{tabular}
}
\end{table}

\bibliographystyle{chicago}
\bibliography{bib}